\documentclass[copyright,creativecommons]{eptcs}
\providecommand{\event}{HCVS 2021}
\usepackage{breakurl} 
\usepackage{underscore} 

\newif\ifdraft\draftfalse

\usepackage[english]{babel}
\usepackage{underscore}
\usepackage{savesym}
\usepackage{amsmath}
\usepackage{mathtools}
\usepackage{stmaryrd}
\usepackage{amssymb}
\usepackage{amsthm}
\usepackage{thmtools}
\usepackage{thm-restate}
\usepackage[inline]{enumitem}
\usepackage{galois}
\usepackage{float}
\usepackage{diagrams}
\usepackage{bussproofs}
\usepackage{multicol}
\usepackage{caption}
\usepackage{subcaption}
\usepackage{qtree}

\DeclareMathAlphabet{\mathcal}{OMS}{cmsy}{m}{n}
\DeclareMathOperator{\mhyphen}{\!-\!}

\newcommand\cons{\mathsf{cons}}
\renewcommand\succ{\mathsf{succ}}
\newcommand\zero{\mathsf{zero}}

\newcommand{\eqA}{\hyperref[eq:homc_A_botfree]{A}}
\newcommand{\eqB}{\hyperref[eq:homc_B_inclusion]{B}}
\newcommand{\eqC}{\hyperref[eq:homc_C_inclusion]{C}}

\newcommand\set[1]{\{#1\}}
\newcommand{\makeset}[1]{\{ #1 \}}
\newcommand\lub{\bigsqcup}
\newcommand\glb{\bigsqcap}
\newcommand{\cTo}{\mathrel{\To_c}}
\newcommand{\mTo}{\mathrel{\To_m}}
\newcommand{\sorts}{\vdash}
\newcommand{\types}{\vdash}

\newcommand{\abra}[1]{\langle #1 \rangle}
\newcommand\anglebra[1]{\langle #1 \rangle}
\newcommand{\To}{\Rightarrow}
\newcommand{\dom}{\mathsf{dom}}
\newcommand\coloneqq{\mathbin{:=}}
\newcommand\Coloneqq{\mathbin{::=}}
\newcommand\defeq{\coloneqq}
\renewcommand{\models}{\vDash}
\renewcommand{\phi}{\varphi}
\newcommand{\vv}[1]{\overline{#1}}
\newcommand{\lfp}{\mathsf{lfp}}
\newcommand{\gfp}{\mathsf{gfp}}

\newcommand\calG{\mathcal{G}}
\newcommand\calL{\mathcal{L}}
\newcommand{\hmng}[1]{\mathcal{H}\llbracket #1 \rrbracket}
\newcommand\ssym{s}
\newcommand\bsym{b}
\newcommand\sem[1]{\llbracket #1 \rrbracket}
\newcommand\lift[1]{{\ulcorner #1 \urcorner}}
\newcommand\dlift[1]{{\ulcorner\!\!\ulcorner #1 \urcorner\!\!\urcorner}}
\newcommand\Rel{\mathsf{Rel}}
\newcommand\id{\mathrm{id}}
\newcommand{\imap}{\mathbf{i}}
\newcommand{\jmap}{\mathbf{j}}
\newcommand{\HORS}{\mathcal{G}}
\newcommand{\nonterms}{\mathcal{N}}
\newcommand{\calR}{\mathcal{R}}
\newcommand{\Prop}{\mathit{Prop}}

\newcommand{\calN}{\mathcal{N}}
\newcommand{\calI}{\mathcal{I}}

\newcommand\arity[1]{\textrm{ar}({#1})}
\newcommand\rsvars{V_\mathrm{RS}}
\newcommand\Trees[1]{\mathcal{T}_{{#1}}}
\newcommand\Sigbot{{\Sigma_\bot}}
\newcommand\alltrees{\Trees{\Sigbot}}
\newcommand{\Incl}{\mathsf{Incl}}
\newcommand{\Image}{{\mathsf{Im}}}

\newcommand{\dmng}[1]{\mathcal{D}\llbracket #1 \rrbracket}
\newcommand{\mmng}[1]{\mathcal{M}\llbracket #1 \rrbracket}

\newcommand{\calP}{\mathcal{P}}

\newcommand{\calM}{\mathcal{M}}

\newcommand{\atheory}{\mathit{Th}}
\newcommand{\aformulas}{\mathit{Fm}}

\newcommand{\aterms}{\mathit{Tm}}
\newcommand{\boolsort}{o}
\newcommand{\intsort}{\mathsf{int}}

\newcommand{\mfunc}{T^\mathcal{M}}
\newcommand{\abs}[2]{\lambda #1.\,#2}

\newcommand{\FormatLabel}[1]{%
  \textsf{(#1)}
}
\newcommand{\RuleName}[2]{%
  \newcommand{#1}{\FormatLabel{#2}}
}
\RuleName{\GVar}{GVar}
\RuleName{\GCst}{GCst}
\RuleName{\GConstraint}{GConstr}
\RuleName{\GAppRel}{GAppR}
\RuleName{\GAppInd}{GAppI}
\RuleName{\GAbs}{GAbs}
\RuleName{\GEx}{GEx}

\makeatletter
\def\term#1{%
    \@term#1 \@empty
}
\def\@term#1 #2{%
   \mathit{#1}
   \ifx #2\@empty\else
    \:\expandafter\@term  %
   \fi
   #2%
}
\makeatother

\savesymbol{comment}

\ifdraft
\usepackage[draft]{commenting}
\else
\usepackage[nompar]{commenting}
\fi

\declareauthor{jj}{JJ}{blue}
\authorcommand{jj}{comment}

\makeatletter
\renewcommand{\comm@todo@mpar}[1]{}
\makeatother

\def\divider{%
  \leavevmode\leaders\hrule height 0.6ex depth \dimexpr0.4pt-0.6ex\hfill%
  \kern0pt%
}

\renewcommand{\OpenCommentBraket}{$\lceil$}
\renewcommand{\ClosedCommentBraket}{$\rfloor$}

\theoremstyle{plain} 
\newtheorem{theorem}{Theorem}[section]
\newtheorem{lemma}[theorem]{Lemma}
\newtheorem{proposition}[theorem]{Proposition}
\newtheorem{corollary}[theorem]{Corollary}

\newtheorem{definition}[theorem]{Definition}
\newtheorem{example}[theorem]{Example}



\title{Reducing Higher-order Recursion Scheme Equivalence \\ to Coinductive Higher-order Constrained Horn Clauses}
\author{Jerome Jochems
\institute{Department of Computer Science\\
University of Bristol\footnote{Most of the work was concluded while the author was a Research Associate at the University of Oxford.}\\
Bristol, UK}
\email{jerome.jochems@bristol.ac.uk}}
\def\titlerunning{Reducing HoRS Equivalence to Coinductive HoCHC}
\def\authorrunning{Jerome Jochems}

\begin{document}

\maketitle

\begin{abstract}
Higher-order constrained Horn clauses (HoCHC) are a semantically-invariant system of higher-order logic modulo theories.
With semi-decidable unsolvability over a semi-decidable background theory, HoCHC is suitable for safety verification.
Less is known about its relation to larger classes of higher-order verification problems.
Motivated by program equivalence, we introduce a coinductive version of HoCHC that enjoys a greatest model property.
We define an encoding of higher-order recursion schemes (HoRS) into HoCHC logic programs.
Correctness of this encoding reduces decidability of the open HoRS equivalence problem -- and, thus, the $\lambda$\textbf{Y}-calculus B\"ohm tree equivalence problem -- to semi-decidability of coinductive HoCHC over a complete and decidable theory of trees.
\end{abstract}


\section{Introduction}
\label{sec:intro}

Cathcart Burn et al.~\cite{BurnOR18} have proposed a promising logical framework for higher-order safety verification.
They frame the search for ``safe'' program invariants as a satisfiability problem for systems of \emph{HoCHC}:
these \emph{higher-order constrained Horn clauses} -- which extend constrained Horn clauses to higher-order logic with constraints from a first-order background theory -- aim to act as a universal setting in which disparate verification algorithms can be compared, independent of application or programming language.

Thanks to its higher-order predicates, the HoCHC fragment expresses certain invariants of higher-order programs quite directly. 
Even so, it retains many of the excellent algorithmic properties to which first-order constrained Horn clauses owe their suitability for first-order model checking~\cite{Bjorner2012,Bjorner2015}. 
Given a semi-decidable background theory, HoCHC unsolvability (unsatisfiability) is semi-decidable~\cite{Pham2018,OngWagner2019}. 

We study the relation between (the logic-program presentation of) HoCHC and \emph{higher-order recursion schemes} (HoRS).
Whilst higher-order \emph{model checking} has grown out of the decidability of HoRS model checking \cite{Ong2006,KobayashiO2009} and flourished, higher-order \emph{program equivalence} is relatively underdeveloped;
decidability of HoRS equivalence is a long-standing open problem~\cite{Clairambault2013,Ong2015}.
Note that the HoRS model checking safety problem can be solved via a decidable higher-order Datalog fragment of HoCHC~\cite{Wagner2019}.

A HoRS of order $n$ is essentially an $n$th-order tree grammar: the trees generated at orders 0, 1 and 2 are regular trees, algebraic trees (i.e.~those generated by context-free tree grammars), and hyperalgebraic trees, respectively~\cite{Courcelle1978}.
These potentially infinite trees generated by HoRS correspond to (abstractions of) computation trees of higher-order functional programs.

Let us consider (deterministic) HoRS $\calG_1$ and HoRS $\calG_2$ in Figure~\ref{fig:HoRS_example} that both generate an infinite tree with the prefix on the right (by unfolding the rewrite rules ad infinitum, starting from $S_1$ and $S_2$, respectively).
To determine whether these HoRS generate the same tree, we define a HoCHC logic program that contains one predicate $R_N$ (of arity $n+1$) for each nonterminal symbol $N$ (of arity $n$) in the input HoRS.
In particular, we want $R_{S_1}$ (resp. $R_{S_2}$) to be the characteristic function of the tree generated from $S_1$ by $\calG_1$ (resp. from $S_2$ by $\calG_2$), so we can query the existence of a tree $t$ such that $R_{S_1}\, t \land R_{S_2}\, t$.

\begin{figure}[t]
\centering
\begin{subfigure}[r]{0.35\textwidth}
\begin{align*}
    S_1 &= G\,\zero\\
    G &= \lambda x.\, \cons\,(\succ\,x)\,(G\,(\succ\,x)) \\[8pt]
        S_2 &= F\,\succ\\
    F &= \lambda \varphi.\, \cons\,(\varphi\, \zero)\,(F\,(B\,\varphi\,\varphi))\\
    B &= \lambda \varphi\,\psi\,x.\, \varphi \, (\psi \, x)
\end{align*}
\vspace*{-13pt}
\caption{Their respective rewrite rules\\}
\end{subfigure}
\begin{subfigure}[c]{0.25\textwidth}
\Tree[.{$\cons$} 
        [.{$\succ$} {$\zero$} ]
        [.{$\cons$} 
            [.{$\succ$} [.{$\succ$} {$\zero$} ] ]
                {\dots} ]]
\caption{A common prefix of the trees generated by $\calG_1$ and $\calG_2$}
\end{subfigure}
\hspace*{0.2cm}\begin{subfigure}[l]{0.36\textwidth}
\vspace*{0.2cm}
\begin{enumerate}[leftmargin=*]
    \item Eliminate divergent $\bot$-labelled leaves from HoRS (Sec~\ref{sec:HoRS_botfree})
    \item Encode HoRS into coinductive HoCHC logic program (Sec~\ref{sec:HoRS_transformation})
    \item Solve the two HoCHC instances from Sec~\ref{sec:homc_procedure} concurrently (outside the scope of this paper)
\end{enumerate}
\vspace*{-0.2cm}
\caption{The ``decision'' procedure for HoRS equivalence, pending semi-decidability}
\label{fig:decision_procedure}
\end{subfigure}
\caption{Example order-1 HoRS $\calG_1$ and order-2 HoRS $\calG_2$ over tree constructors $\makeset{\cons,\succ,\zero}$}
\label{fig:HoRS_example}
\vspace*{-0.3cm}
\end{figure}

\emph{Encoding} HoRS into HoCHC (in ``continuation-passing style'') is natural, as this program shows:
\begin{align*} 
    R_{S_1}
        &= \lambda r.\,\exists r_1.\, R_G\,r_1\,r \land (\zero = r_1)\\
    R_G
        &= \lambda x'\,r.\, \exists r_1\, r_2\, r_3.\, (\cons\,r_1\,r_2=r) \land (\succ\,x' =r_1) \land R_G\,r_3\,r_2 \land (\succ\,x' =r_3) \\
    R_{S_2} &= \lambda r.\, R_F\,(\lambda y\,r'.\, \succ\,y = r')\,r \\
    R_F &= \lambda \varphi'\,r.\, \exists r_1\,r_2\,r_3.\, (\cons\,r_1\,r_2 = r) \land \varphi'\,r_3\,r_1 \land (\zero = r_3) \land R_F\,(\lambda y\,r'.\, R_B \, \varphi' \, \varphi'\, y \, r')\,r_2 \\
    R_B &= \lambda \varphi'\,\psi'\,x'\,r.\, \exists r_1\,r_2.\, \varphi' \,r_1\,r \land \psi'\,r_2\,r_1 \land (x' = r_2)
\end{align*}
Each HoRS subterm is represented by a conjunct whose arguments are $r_i$ bounded by a subsequent conjunct (if a tree) or are inlined (if higher type).
Unfortunately though, this HoCHC logic program has no natural inductive interpretation.
The empty assignment is a model (a fixpoint) of the program, because the program contains no ``base cases'' to break out of the recursion.
In fact, there does not exist a HoCHC term such that the characteristic function of an infinite tree arises in its least fixpoint.

To tackle this disparity between HoCHC and HoRS, we define a \emph{coinductive} HoCHC framework that enjoys a greatest model property (under the monotone interpretation).
We interpret the above clauses coinductively over a complete and decidable background theory of trees first introduced by Maher~\cite{Maher1988}.

In our example, $R_{S_1}$ is assigned the characteristic function of the tree generated by $\calG_1$ in the greatest model (and $R_{S_2}$ of $\calG_2$).
We can (independently) query the existence of two identical and two distinct trees $t_1$ and $t_2$ such that $R_{S_1}\, t_1 \land R_{S_2}\, t_2$.
Only two distinct such trees exist;
after the common prefix in Figure~\ref{fig:HoRS_example}, the trees deviate.
The left children of $\cons$ have shape $\succ^n\; \zero$ in $t_1$ and $\succ^{2^{n-1}}\, \zero$ in $t_2$, for $n \geq 1$.

This new framework allows us to characterise HoRS in HoCHC logic programs and, thus, reduce decidability of the HoRS equivalence problem to semi-decidability of coinductive HoCHC over a decidable background theory.
This has implications for the $\lambda\mathbf{Y}$-calculus B\"ohm tree equivalence problem~\cite{Clairambault2013}, which asks whether the B\"ohm trees of two given $\lambda\mathbf{Y}$-terms are equal; this problem is recursively equivalent to the HoRS equivalence problem and can also be reduced to semi-decidability of coinductive HoCHC.

\paragraph{Contributions.}

The (open) HoRS equivalence problem asks whether two given deterministic HoRS generate the same tree.
We prove that decidability of this problem can be reduced to the (open) semi-decidability of coinductive HoCHC over a decidable background theory (see Figure~\ref{fig:decision_procedure}).

First, we prove that there is an algorithm that, given a HoRS (which may contain ``diverging'' $\bot$-labelled nodes), returns its $\bot$-free transform -- i.e.~a HoRS that generates the same tree, except that every $\bot$-labelled node is replaced by the infinite linear tree $\bsym \ (\bsym \ (\bsym \cdots))$.
The proof appeals to the logical reflectivity of HoRS with respect to properties definable in monadic second-order logic, in the sense of \cite{DBLP:conf/lics/BroadbentCOS10}.
Notice that two HoRS are equivalent if and only if their respective $\bot$-free transforms are equivalent.

Next, we exhibit a natural encoding of HoRS into constrained logic programs, with the sort of individuals interpreted as the set of finite and infinite trees.
Given two $\bot$-free HoRS, we define two instances of the coinductive HoCHC problem, call them positive and negative.
We use Maher's first-order theory of equations of finite and infinite trees \cite{Maher1988}, which is complete and decidable, as the background theory.
The positive and the negative problem instance share a logic program: the union of the respective HoRS-to-HoCHC encodings.
The goal formulas of the problem instances are so designed that the two input HoRS are equivalent (resp.~inequivalent) iff the positive (resp.~negative) instance is solvable.
Provided that the resulting coinductive HoCHC instances are semi-decidable, we obtain two semi-decision procedures, one for checking equivalence of the input HoRS and one for inequivalence.
A decision procedure for the equivalence of the input HoRS could then be obtained by dovetailing the two semi-decision procedures.

\paragraph{Outline.}

Building on Cathcart Burn et al.'s (inductive) HoCHC~\cite{BurnOR18}, we introduce coinductive HoCHC in Section~\ref{sec:HoRS_coinductive_HoCHC}.
We define HoRS and their denotational semantics in Section~\ref{sec:hors}, where we also prove the existence of an algorithm that generates the $\bot$-free transform of HoRS.
We encode HoRS into constrained logic programs in Section~\ref{sec:transformation}.
Section~\ref{sec:HoRS_equivalence} shows how to use these HoRS-to-HoCHC encodings to reduce decidability of the HoRS equivalence problem to semi-decidability of coinductive HoCHC over Maher's complete and decidable theory of trees.
Finally, we consider implications and related work in Section~\ref{sec:conclusion}.


\section{Preliminaries}
\label{sec:prelims}


\subsection{Higher-order constrained Horn clauses}
\label{sec:coinductive_HoCHC}
\label{sec:HoRS_coinductive_HoCHC}

Following \cite{BurnOR18}, we work in higher-order logic
presented as a typed (sorted) lambda calculus.
We follow their (monotone logic-program) definitions until we introduce the coinductive HoCHC decision problem.

\subsubsection{Syntax}
\paragraph{Sorts.}
Given a sort $\iota$ of individuals (for example $\intsort$), and a sort $\boolsort$ of (boolean) truth values, \emph{sorts} are just the simple types generated by 
$\sigma \Coloneqq \iota \mid \boolsort \mid \sigma\to\sigma$. 
\emph{Relational sorts} (typically denoted by $\rho$) have the following restricted form: $\rho \Coloneqq \boolsort \mid \iota \to \rho \mid \rho \to \rho$.

\paragraph{Background theory.}
Assume a fixed, first-order language over a first-order signature, consisting of distinguished subsets of first-order terms $\aterms$ and first-order formulas (or \emph{constraints}) $(\phi \in)$ $\aformulas$, and a first-order theory $\atheory$ in which to interpret those. 
We fix a \emph{standard model} $A$ of $\atheory$ we often leave implicit.
We refer to this first-order language as the \emph{constraint language}, and $\atheory$ as the \emph{background theory}.

\paragraph{Goal terms.}
The class of well-sorted \emph{goal terms} $\Delta \sorts G : \rho$ is given by the sorting judgements defined by the following rules, where $\sigma$ stands for the sort of individuals $\iota$ or some relational sort.

\noindent
\begin{center}
\hspace{35pt}\begin{minipage}{.4\linewidth}
\begin{prooftree}
    \AxiomC{\vphantom{$\Delta$}}
      \RightLabel{$\Delta \sorts \phi : o \in \aformulas$}
      \LeftLabel{\GConstraint}
      \UnaryInfC{$\Delta \sorts \phi : o$}
\end{prooftree}
\end{minipage}\begin{minipage}{.45\linewidth}
\begin{prooftree}
    \AxiomC{}
      \LeftLabel{\GVar}
      \UnaryInfC{$\Delta_1, x:\rho, \Delta_2 \sorts x : \rho$}
\end{prooftree}
\end{minipage}\\\begin{minipage}{.40\linewidth}
\begin{prooftree}
    \AxiomC{$\Delta \sorts G : \boolsort$}
      \AxiomC{$\Delta \sorts H : \boolsort$}
      \LeftLabel{\GCst}
      \RightLabel{$* \in \makeset{\wedge, \vee} $}
      \BinaryInfC{$\Delta \sorts G * H : \boolsort$}
\end{prooftree}
\end{minipage}\hspace{30pt}
\begin{minipage}{.40\linewidth}
\begin{prooftree}
    \AxiomC{$\Delta,x:\sigma \sorts G : \boolsort$}
      \LeftLabel{\GEx}
      \RightLabel{$\sigma = \iota \textrm{ or } \rho$}
      \UnaryInfC{$\Delta \sorts \exists x\hspace{-2pt}:\hspace{-2pt}\sigma.G  : \boolsort$}
\end{prooftree}
\end{minipage}\\\hspace{-15pt}\begin{minipage}{.4\linewidth}
\begin{prooftree}
    \AxiomC{$\Delta \sorts G : \iota \to \rho$}
      \RightLabel{$\Delta \sorts N : \iota \in \aterms$}
      \LeftLabel{\GAppInd}
      \UnaryInfC{$\Delta \sorts \term{G N} : \rho$}
\end{prooftree}
\end{minipage}\hspace{27pt}
\begin{minipage}{.45\linewidth}
\begin{prooftree}
    \AxiomC{$\Delta \sorts G : \rho_1 \to \rho_2$}
      \AxiomC{$\Delta \sorts H : \rho_1$}
      \LeftLabel{\GAppRel}
      \BinaryInfC{$\Delta \sorts \term{G H} : \rho_2$}
\end{prooftree}
\end{minipage}\\\vspace{-8pt}
\begin{prooftree}
    \AxiomC{$\Delta, x:\sigma \sorts G : \rho$}
      \LeftLabel{\GAbs}
      \RightLabel{$x \notin \dom(\Delta)$}
      \UnaryInfC{$\Delta \sorts \abs{x}{G} : \sigma \to \rho$}
\end{prooftree}
\end{center}\vspace*{4pt}
From now on, we use $G$ and $H$ (and variants thereof) to stand for arbitrary goal terms and disambiguate as necessary, and use uppercase $R$ to stand for \emph{relational variables} (i.e.~variables of a relational sort).

\paragraph{Constrained logic program.}
A higher-order constrained \emph{logic program}, $P$, over a sort environment $\Delta = \{R_1:\rho_1,\ldots,R_k:\rho_k\}$ is a finite system of (mutually) recursive definitions of shape $R_i :\rho_i = {G_i}$ for some goal term $G_i$.
Such a program is well sorted when $\Delta \types {G_i} : \rho_i$, for each $1 \leq i \leq k$. 
Since each $R_i$ is distinct, we will sometimes regard a program $P$ as a finite map from variables to terms, defined so that $P(R_i) = G_i$. 
We write $\types P : \Delta$ to mean that $P$ is a well-sorted program over $\Delta$. 

\subsubsection{Semantics}
\label{sec:hochc_interpretations}

Motivated by the fact that unlike its first-order counterpart, HoCHC has no least model property for standard semantics, Cathcart Burn et al.~consider an equivalent \emph{monotone semantics} that \emph{does} have a least model property.
This interpretation suits us too, because it also has a greatest model property.

\paragraph{Monotone sort frame.}
We define the \emph{monotone sort frame} $\mmng{-}$ over the domain $A_\iota$ of the background theory recursively by:
\[ \mmng{\iota} := A_\iota \qquad \mmng{o} := \mathbb{B}\defeq\{0 \leq 1\}
  \qquad \mmng{\rho_1 \to \rho_2} := \mmng{\rho_1} \mTo
  \mmng{\rho_2} \]
where $X \mTo Y$
is the monotone function space between $X$ and $Y$ w.r.t.~a partial
ordering $\sqsubseteq$;
this partial order is the discrete ordering on $A_\iota$,
satisfies $0\sqsubseteq 1$, and is lifted to higher sorts in a pointwise manner.
It is easy to see that each $\mmng{\rho}$ is a complete lattice.

We extend this ordering to sort environments with $\mmng{\Delta} := \prod x \in \mathsf{dom}(\Delta).\, \mmng{\Delta(x)}$, pointwise over its elements, i.e., for $\beta,\gamma \in \mmng{\Delta}$, $\beta \sqsubseteq \gamma$ if and only if $\beta(x) \sqsubseteq \gamma(x)$ for all $x:\rho \in \Delta$.

\paragraph{Denotation of goal terms.}
The meaning $\mmng{\Delta \sorts G : \rho}: \mmng{\Delta} \to \mmng{\rho}$ of a goal term $\Delta\sorts G:\rho$ is defined as follows, for $\beta \in \mmng{\Delta}$:
\begin{align*}
  \mmng{\Delta\sorts\phi:o}(\beta)
                                         &\defeq\atheory\llbracket\phi\rrbracket(\beta)\\
  \mmng{\Delta_1,x:\rho,\Delta_2\sorts x:\rho}(\beta) &\defeq\beta(x)\\
  \mmng{\Delta\sorts G\land H:o}(\beta) &\defeq\min\{\mmng{\Delta\sorts G: o}(\beta), \mmng{\Delta\sorts H: o}(\beta)\}\\
  \mmng{\Delta\sorts G\lor H:o}(\beta) &\defeq\max\{\mmng{\Delta\sorts G: o}(\beta), \mmng{\Delta\sorts H: o}(\beta)\}\\
  \mmng{\Delta\sorts\exists x:\sigma\ldotp  G:o}(\beta)
                                         &\defeq\max\{\mmng{\Delta,x:\sigma\sorts
                                           G: o}(\beta[x\mapsto
                                           x'])\mid x'\in\mmng\sigma\}\\
  \mmng{\Delta\sorts\lambda x:\sigma\ldotp G:\sigma\to\rho}(\beta)
                                         &\defeq\lambda
                                           x'\in\mmng\sigma\ldotp
                                           \mmng{\Delta,x:\sigma\sorts
                                           G:\rho}(\beta[x\mapsto
                                           x'])\\
  \mmng{\Delta\sorts G\,H:\rho_2}(\beta) &\defeq\mmng{\Delta\sorts
                                          G:\rho_1\to\rho_2}(\beta)(
                                          \mmng{\Delta\sorts
                                            H:\rho_1}\llbracket
                                            H\rrbracket(\beta))\\
    \mmng{\Delta\sorts G\,N:\rho}(\beta) &\defeq\mmng{\Delta\sorts
                                          G:\iota\to\rho}(\beta)(
                                          \atheory\llbracket N\rrbracket(\beta))
\end{align*}
In the above, $\min$ and $\max$ denote the greatest lower bound and the least upper bound, resp., within the complete lattice of booleans $\mmng{o}$, and $\atheory\llbracket - \rrbracket(\beta)$ denotes the interpretation of $-$ in the (standard) model of the background theory.
We write $A,\beta \vDash G:o$, i.e.~$A,\beta$ \emph{satisfies} $G$, just if $\mmng{\Delta \vdash G:o}(\beta)=1$.

\paragraph{One-step consequence operator.}
Logic programs give rise to an endofunction $T_{P:\Delta}^{\mathcal M}: \mmng{\Delta} \to \mmng{\Delta}$, defined by $T_{P:\Delta}^{\mathcal M}\,(\beta)\,(R_i) \defeq\mmng{\Delta \sorts P(R_i):\Delta(R_i)}(\beta)$,
called the \emph{one-step consequence operator}.
We call $\beta$ a \emph{model} of $\vdash P:\Delta$, written $A,\beta \vDash P$ for model $A$ of the background theory, just if $\beta = T_{P:\Delta}^{\mathcal M}(\beta)$.

\subsubsection{Coinductive decision problem}
\label{sec:coind_coind}

\begin{definition}[Coinductive HoCHC problem]
A coinductive HoCHC problem $\abra{\Delta,P,G}$, where 
$\Delta$ is a sorting of relational variables,
$\sorts P : \Delta$ is a constrained logic program, and
$\Delta \sorts G : \boolsort$ is a constrained goal formula,
is \emph{solvable} just if, for the standard model $A$ of the background theory $\atheory$, there exists a valuation $\beta$ of the variables in $\Delta$ such that $A,\beta \models P$ and $A,\beta \models G$.
\end{definition}

Note that the problem triple $\abra{\Delta,P,G}$ is identical to its inductive counterpart, but the definition of ``solvability'' differs;
in the original HoCHC problem, solvability requires the existence of prefixed point $\beta$ of $P$ such that $A,\beta \not\models G$.
The background theory $\atheory$ could be any first-order theory -- in which the constraints in $P$ and $G$ can be interpreted -- but in Section~\ref{sec:HoRS_equivalence} we fix a specific theory of trees.

Interpreting satisfaction w.r.t.~standard and monotone semantics gives rise to two distinct but equivalent HoCHC decision problems.
A dual argument to Cathcart Burn et al.'s  Lemma~5 for inductive HoCHC~\cite{BurnOR18} shows that a coinductive HoCHC problem is \emph{solvable under the standard interpretation} iff it is \emph{solvable under the monotone interpretation.}
We consider the monotone interpretation.

By the Knaster-Tarski theorem and $\mmng{\Delta}$ being a complete lattice for relational $\Delta$, the set of fixpoints of the monotone one-step consequence operator $\mfunc_{P:\Delta}$ forms a complete lattice. This guarantees the existence of a greatest fixpoint of $\mfunc_{P:\Delta}$.
Thus, monotone HoCHC enjoys a greatest model property.

\begin{theorem}[Greatest model property for monotone HoCHC]\label{thm:homc_coind_hochc_gmp}
Under the monotone interpretation, HoCHC definite clauses possess greatest models.
\end{theorem}

Thus, a greatest model witnesses the solvability of a coinductive HoCHC problem in the monotone setting, like a least model witnesses solvability for traditional HoCHC.
Instead of building up a least model from the least valuation, we start with the greatest valuation and work our way down.
Intuitively, we are taking the \emph{backwards closure} of our logic program.

\begin{theorem}
A coinductive HoCHC problem $\abra{\Delta,P,G}$ is solvable under the monotone interpretation if and only if $\mmng{G}(M_P) = 1$ for greatest model $M_P$ of $P$.
\end{theorem}
\begin{proof}
Recall that a coinductive HoCHC problem $\abra{\Delta,P,G}$ is solvable iff, for the standard model $A$ of the background theory $\atheory$, 
there exists a valuation $\beta$ of the variables in $\Delta$ such that $A,\beta \models P$ and $A,\beta \models G$.

Clearly, $\mmng{G}(M_P) = 1$ for greatest model $M_P$ of $P$ implies solvability of $\abra{\Delta,P,G}$.
For the converse, let $\beta \in \mmng{\Delta}$ be a valuation such that $A,\beta \vDash P$ and $A,\beta \vDash G$, for standard model $A$ of the background theory.
By Knaster-Tarski, $\beta \sqsubseteq M_P$.
By monotonicity, $A,\beta \vDash G$ implies $A,M_P \vDash G$, as required.
\end{proof}

HoCHC and coinductive HoCHC are not equivalent in the sense that the standard and the monotone interpretation are, as the following example demonstrates;
it is not the case that a HoCHC problem $\abra{\Delta,P,G}$ is solvable if and only if the corresponding coinductive HoCHC problem $\abra{\Delta,P,G}$ is solvable.

\begin{example}\label{ex:coinductive_encoding}
Consider the HoCHC triples $\calP_a = \anglebra{\makeset{R_S},P,R_S\,a^\omega}$ and $\calP_b = \anglebra{\makeset{R_S},P,R_S\,b^\omega}$, where $P$ consists of
$R_S = \lambda r.\, \left(\exists r_1.\, (a\,r_1 = r) \land R_S\,r_1 \right)$
and $a^\omega$ is the infinite unary tree of only $a$s (and $b^\omega$ of $b$s).
The logic program $P$ has two models;
relational variable $R_S$ corresponds to the empty set in least model $M_\emptyset$, and to the singleton set $\makeset{a^\omega}$ in greatest model $M_{\makeset{a^\omega}}$.

When we consider these HoCHC triples as inductive HoCHC problems, we find both $\calP_a$ and $\calP_b$ solvable, since $M_\emptyset$ is a witness to the refutation of both goal clauses.
If we consider them as coinductive problems, however, $\calP_a$ is solvable, while $\calP_b$ is unsolvable. 
This follows from $M_{\makeset{a^\omega}}$ witnessing the satisfiability of goal clause $R_S\,a^\omega$, while neither $M_{\makeset{a^\omega}}$ nor $M_{\emptyset}$ witnesses the satisfiability of $R_S\,b^\omega$.
\end{example}

Coinductive higher-order constrained Horn clauses allow us to reason about programs with datatypes inhabited by infinite objects, notably the (potentially) infinite trees generated by HoRS. 
These clauses are not merely an academic indulgence, though.
There is a tradition of coinduction and corecursion in logic programming (see e.g.~\cite{Jaffar1986,Gupta2007,Simon2007,Komendantskaya2017}).
This is hardly surprising, given that some well-formed logic programs do not have natural inductive interpretations, as we have seen in the introduction. Examples of other infinite data types that arise in practice include infinite lists and streams.


\subsection{Higher-order recursion schemes}
\label{sec:hors}
\label{sec:PRE_HoMC}

We fix a ranked alphabet $\Sigma$ of tree constructors and write $\Sigma_\bot = \Sigma \cup \makeset{\bot}$. 
The set of all finite and infinite $\Sigma_\bot$-labelled trees, written $\Trees{\Sigbot}$, 
is a pointed poset with least element $\bot$ over the subtree ordering $\sqsubseteq$, which is
the least partial order such that
$C[\bot] \sqsubseteq C[t]$ for every tree context $C[\mbox{\textvisiblespace}] \in \Trees{\makeset{\mbox{\textvisiblespace}} \cup \Sigbot}$ and $t \in \Trees{\Sigbot}$.

Let $\calG = \anglebra{\calN,\Sigma,\calR,S}$ be a (deterministic) higher-order recursion scheme (HoRS).
That is, $\calN$ maps a nonterminal symbol to its sort, $\Sigma$ maps a terminal symbol to its sort, $S \in \calN$ is the designated start symbol, and
there exists one rewrite rule in $\calR$ for each $F \in \calN$ such that $\calR(F) = \lambda x_1 \dots x_n.\,t$,
where $t:\iota$ is an applicative term over $\calN \cup \Sigma \cup \makeset{x_1,\dots,x_n}$ for some distinct $x_1,\dots,x_n {\in \rsvars}$ drawn from a finite set of recursion scheme variables $\rsvars$.

The \emph{HoRS equivalence problem} asks whether two given HoRS generate the same tree (i.e.~have the same semantics).
Decidability of this problem is perhaps the best known and most challenging open problem in higher-order model checking.

\subsubsection{Denotational semantics}
\label{sec:HoRS_semantics}

The meaning of a HoRS can be given by a number of different formalisms.
We introduce an infinite ``Herbrand'' interpretation that treats the rewrite rules as definitional equality in the style of a HoCHC logic program.
Models are built incrementally from the smallest tree $\bot$.

Our interpretation of HoRS is Herbrandesque in that constants and function symbols are assigned very simple meanings.
However, unlike typical Herbrand models, our models may contain infinite terms.

Let us define an interpretation of the sorts over $\iota$ (i.e.~the sorts of HoRS terms):
\[ \hmng{\iota} \defeq \abra{\alltrees,\sqsubseteq} \qquad \hmng{\sigma_1 \to \sigma_2} \defeq \hmng{\sigma_1} \Rightarrow_c \hmng{\sigma_2}\]
where $X \cTo Y$ is the continuous function space between directed-complete partial orders (dcpos) $X$ and $Y$, ordered pointwise with respect to subtree ordering $\sqsubseteq$ on $\hmng{\iota}$.

Given the environment $\Gamma = \makeset{x_1:\tau_1, \dots, x_k:\tau_k}$, let $\calN'$ denote the extended environment $\calN,\Gamma \coloneqq \calN \cup \Gamma$, which we view as a sort function whose domain is $\dom(\calN) \cup \makeset{x_1,\dots,x_k}$, mapping each symbol to its sort.
Set $\hmng{\calN'}: {\prod} x \in \mathsf{dom}(\calN').\,\hmng{\calN'(x)}$ with typical element $\alpha$.
Define
\[ \hmng{\calN' \vdash t:\sigma}: \hmng{\calN'} \Rightarrow \hmng{\sigma} \]
by cases and recursion on syntax: 
\begin{align*}
    \hmng{\calN' \vdash x: \calN'(x)}(\alpha) &= \alpha(x)\\
    \hmng{\calN' \vdash f: \iota^{\arity{f}} \to \iota}(\alpha) &= \widehat{F_f}\\
    \hmng{\calN' \vdash s\,t: \tau}(\alpha) &= \hmng{\calN' \vdash s:\sigma \to \tau}(\alpha) (\hmng{\calN' \vdash t: \sigma}(\alpha))\\
    \hmng{\calN' \vdash \lambda x : \sigma_1.\,t:\sigma_2}(\alpha) &= \lambda v \in \hmng{\sigma_1}.\,\hmng{\calN', x:\sigma_1 \vdash t}(\alpha[x \mapsto v])
\end{align*}
for $\widehat{F_f} \in \hmng{\iota^{\arity{f}} \to \iota}$ the usual Herbrand interpretation of $f:\iota^{\arity{f}} \to \iota \in \Sigma$. 
We define $\hmng{\calG}_{\calN'}: \hmng{\calN'} \To \hmng{\calN'}$ pointwise by 
$\hmng{\calG}_{\calN'}\,(\alpha)\,(x) \defeq \hmng{\calN' \vdash \calR(x):\calN'(x)}(\alpha)$.

\begin{lemma}
\label{lem:H_continuous}
$\hmng{\calG}_{\calN} : \hmng{\calN} \Rightarrow \hmng{\calN}$ is continuous for all deterministic HoRS $\calG = \abra{\calN,\Sigma,\calR,S}$.
\end{lemma}

We define the (denotational) semantics of $\calG$ as $\lfp(\hmng{\calG}_{\calN})$, the least fixpoint of the continuous endofunction $\hmng{\calG}_{\calN}$, which is well-defined by Kleene's Theorem.
Henceforth we write $\sem{\calG} \defeq \lfp(\hmng{\calG}_{\calN}) (S)$, where $S \in \calN$ is the start symbol of $\calG$.
Note that $\sem{\calG}$ is the $\Sigbot$-labelled tree generated by $\calG$.

\subsubsection{Computability of \texorpdfstring{$\bot$}{bottom}-free transform of HoRS}
\label{sec:HoRS_botfree}

Intuitively, eliminating $\bot$ from $\sem{\calG}$ allow us to distinguish ``unfinished'' trees from ``finished'' (but diverging) trees in e.g.~the proofs in Section~\ref{sec:correctness}.

As usual, let $\Sigbot$ be a finite ranked alphabet $\Sigma$ extended with (nullary) $\bot$.
Let $\bsym:\iota \to \iota \notin \Sigbot$ (for ``bottom'') be a fresh terminal symbol.

\begin{definition}
Given a $\Sigbot$-labelled tree, its \emph{$\bot$-free conversion} is obtained by replacing every $\bot$-labelled node by the infinite linear tree $\bsym \ (\bsym \ 
 (\bsym \ \dots))$.
\end{definition}

\begin{lemma}[Computability of $\bot$-free transform of HoRS]
\label{lem:bot-free-transform}
There is an algorithm that, given a HoRS $\calG$, returns a HoRS -- call it the \emph{$\bot$-free transform} of $\calG$ -- that generates the $\bot$-free conversion of $\sem{\calG}$. 
\end{lemma}

For clarity, we \emph{convert} trees, but \emph{transform} HoRS (their generators).

It is clear from freshness of $\bsym : \iota \to \iota \notin \Sigbot$ that the following holds.

\begin{proposition}
HoRS are equivalent if and only if their respective $\bot$-free transforms are equivalent.
\end{proposition}

Before we present a three-stage algorithm to transform a HoRS $\calG = \abra{\calN,\Sigma,\calR,S}$ to its $\bot$-free transform and an example in Figure~\ref{fig:botfree_example}, we require some background on logical reflection.

\subsubsection*{Logical reflection of HoRS}

Let $\mathfrak{R}$ be a class of generators of $\Sigma$-labelled trees, and $\calL$ be a set of correctness properties of these trees. 
Define the ranked alphabet $\Sigma' := \set{\underline{f} : \sigma \mid f : \sigma \in \Sigma}$, which is a copy of $\Sigma$.
Given a generator $\calG \in \mathfrak{R}$ and property $\phi \in \calL$, we say that $\calG_\phi$ is a \emph{$\phi$-reflection} of $\calG$ just if 
\begin{enumerate}
\item $\calG$ and $\calG_\phi$ generate the same underlying tree, and 
\item if node $\alpha$ of $\sem{\calG}$ has label $f$, then node $\alpha$ of $\sem{\calG_\phi}$ is labelled $\underline{f}$ if $\alpha$ satisfies $\phi$ and $f$ otherwise. 
\end{enumerate}
We say that $\mathfrak{R}$ is \emph{reflective w.r.t.~$\calL$} just if there is an algorithm that transforms a given pair $\abra{\calG, \phi}$ to $\calG_\phi$.

\begin{theorem}[\cite{DBLP:conf/lics/BroadbentCOS10}]\label{thm:HoRS_reflective}
HoRS are reflective w.r.t.~modal $\mu$-calculus and monadic second-order logic.
\end{theorem}

\paragraph{Stage 1: From $\Sigbot$-labelling $\calG$ to $\Sigma \cup \set{\bsym}$-labelling $\calG_1$.}

The input HoRS $\calG$ is first transformed to a \emph{$\bsym$-productive} counterpart $\calG_1 \coloneqq \abra{\calN,\Sigma \cup \makeset{\bsym},\calR',S}$. 
The idea is that in the potentially infinite process of generating the tree $\sem{\calG_1}$ from the start nonterminal $S$ by leftmost-outermost rewriting, each rewriting step is witnessed by either a terminal symbol from $\Sigma$ or by $\bsym$.
The set $\calR'$ of rewrite rules of $\calG_1$ is defined as follows.
For every $F:\sigma_1 \to \dots \to \sigma_n \to \iota \in \calN$:
\begin{itemize}
    \item if $\calR(F) = \lambda x_1 \dots x_n.\,f \, t_1 \dots t_m$ for some $f \in \Sigma$, then $\calR'(F) \defeq \calR(F)$
    \item if $\calR(F) = \lambda x_1 \dots x_n.\,\$ \, t_1 \dots t_m$ for $\$ \in \calN \cup \rsvars$, then $\calR'(F) \defeq \lambda x_1 \dots x_n.\,\bsym\,(\$ \, t_1 \dots t_m)$
\end{itemize}
Notice that the tree $\sem{\calG_1}$, by construction, does not have any $\bot$-labelled nodes.
Intuitively we can get $\sem{\calG}$ back from $\sem{\calG_1}$ by erasing finite $\bsym^\ast$, and replacing infinite $\bsym^\omega$ by $\bot$.

\paragraph{Stage 2: From $\Sigma \cup \set{\bsym}$-labelling $\calG_1$ to $\Sigma \cup \set{\bsym, \ssym}$-labelling $\calG_2$.} 

We define a modal $\mu$-calculus formula:
\[
    \varphi 
        \defeq p_\bsym \land \mu X. \left(\bigvee_{f \in \Sigma} \diamond_1 p_f \lor \diamond_1 X \right)
\]
where $p_\bsym$ (resp.~$p_f$ for $f \in \Sigma$) is a propositional variable that denotes that a node is labelled with $\bsym$ (resp.~$f \in \Sigma$).
Refer to \cite{Bradfield2001} for the syntax and semantics of the modal $\mu$-calculus.
Note that this formula holds for $\bsym$-labelled nodes that are not ``part of some infinite $\bsym^\omega$''.

Let $\ssym:\iota \to \iota \not\in \Sigbot$ (for ``step'') be another fresh arity-1 terminal symbol.
Consider the following operation on $\Sigma \cup \makeset{\bsym}$-labelled trees.
\begin{quote} 
For every node $\alpha$, if $\alpha \vDash \phi$ then rewrite the label at $\alpha$ to $\ssym$, otherwise do nothing.
\end{quote}
This operation leaves exactly those occurrences of $\bsym$ in some infinite $\bsym^\omega$ (which witnesses $\bot$) intact, while rewriting finite paths $\bsym^\ast$ to $\ssym^\ast$.
We call this operation \emph{$\bsym$-to-$\ssym$ conversion}.

Thanks to Theorem~\ref{thm:HoRS_reflective}, the main result in \cite{DBLP:conf/lics/BroadbentCOS10}, there is an algorithm that, given $\calG_1$, returns a HoRS $\calG_2$ over $\Sigma \cup \makeset{\bsym,\ssym}$ that generates the $\bsym$-to-$\ssym$ conversion of $\sem{\calG_1}$.
In the language of~\cite{DBLP:conf/lics/BroadbentCOS10}, $\calG_2$ is the \emph{$\phi$-reflection} of $\calG_1$ where $\phi$ (above) is a property definable in the modal $\mu$-calculus.

\paragraph{Stage 3: From $\Sigma \cup \set{\bsym, \ssym}$-labelling $\calG_2$ to $\Sigma \cup \set{\bsym}$-labelling $\calG_3$.}

Although the tree $\sem{\calG_2}$ does not have infinite paths exclusively labelled by $\ssym$,
it may still have nodes labelled by $\ssym$.
We construct a $\Sigma \cup \set{\bsym}$-labelling HoRS $\calG_3$ that generates the tree $\sem{\calG_2}$ but with these remaining $\ssym$-labelled nodes cut out, which is easily achieved by replacing every occurrence of the terminal symbol $\ssym$ in the rewrite rules of $\calG_2$ by the identity nonterminal $I$.
To be precise, if $\calR_2$ is the set of rewrite rules of $\calG_2$, then the resultant HoRS
\[
\calG_3
\coloneqq
\anglebra{
\calN \cup \set{I}, 
\Sigma \cup \set{\bsym}, 
\set{F = \lambda \overline x . t[I / \ssym] \mid 
F = \lambda \overline x . t \in \calR_2} 
\cup \set{I = \lambda x . x},
S
}
\]
is the $\bot$-free transform of the input HoRS $\calG$.
\vspace*{-11pt}
\begin{figure}[h]
\centering
\raisebox{0.4cm}{
\begin{subfigure}[l]{0.25\textwidth}
\begin{align*}
    S &= \textcolor{red}{b\,(}F\,\zero\textcolor{red}{)}\\
    F &= \lambda x.\, \cons\,(G\,x)\,(F\,(\succ\,x)) \\
    G &= \lambda x.\, \textcolor{red}{b\,(}G\, (\succ \, x)\textcolor{red}{)}
\end{align*}
\end{subfigure}}
\hspace{0.4cm}
\raisebox{1cm}{
\scalebox{.8}{
\Tree[.{$\cons$} {$\bot$} 
        [.{$\cons$} {$\bot$} {\dots} ]]
}
$\mapsto$ 
\scalebox{.8}{
\Tree[.{$\textcolor{red}{b}$} 
		[.{$\cons$} {$\textcolor{red}{b}^\omega$} 
        [.{$\cons$} {$\textcolor{red}{b}^\omega$} {\dots} ]]]
}
$\mapsto$ 
\scalebox{.8}{
\Tree[.{$\textcolor{blue}{s}$} 
		[.{$\cons$} {$\textcolor{red}{b}^\omega$} 
        [.{$\cons$} {$\textcolor{red}{b}^\omega$} {\dots} ]]]
}
$\mapsto$ 
\scalebox{.8}{
\Tree[.{$\cons$} {$\textcolor{red}{b}^\omega$} 
        [.{$\cons$} {$\textcolor{red}{b}^\omega$} {\dots} ]]
}
}
\vspace*{-4pt}
\caption{Conversion $\sem{\calG} \mapsto \sem{\calG_1} \mapsto \sem{\calG_2} \mapsto \sem{\calG_3}$, with HoRS $\calG$ on the left (w/o $\textcolor{red}{b}$s; $\calG_1$ with $\textcolor{red}{b}$s)}
\label{fig:botfree_example}
\end{figure}
\vspace*{-20pt}


\section{Encoding HoRS-to-HoCHC logic program}
\label{sec:HoRS_transformation}
\label{sec:transformation}

In this section, we encode a (deterministic) HoRS $\calG = \abra{\calN,\Sigma,\calR,S}$ into a HoCHC logic program that captures the meaning of the HoRS under the coinductive monotone interpretation.
We assume $\sem{\calG}$ does not contain $\bot$-labelled nodes, which is WLOG by Lemma~\ref{lem:bot-free-transform}.

We define the \emph{HoRS-to-HoCHC encoding $\sorts P_\HORS : \Delta_\HORS$ of HoRS $\calG$} over the coinductive monotone HoCHC interpretation where $\mmng{\iota}$ is interpreted as the underlying set of $\hmng{\iota}$, which is the set of finite and infinite trees over $\Sigbot$.
The constrained logic program $\sorts P_\HORS : \Delta_\HORS$ is defined by:
\[
\Delta_\HORS 
\defeq
\left\{
     R_{F} : \Rel^+(\sigma) \mid F:\sigma \in \calN
\right\} 
\qquad
P_{\HORS} 
\defeq
\left\{
    R_{F} :\Rel^+(\sigma) = \lift{\calR(F)} \mid F:\sigma \in \calN
\right\}
\]
where 
\[
    \Rel^-(\iota) \defeq \iota \qquad 
    \Rel^+(\iota) \defeq\iota \to o \qquad 
    \Rel^-(\sigma_1 \to \sigma_2) \defeq \Rel^+(\sigma_1 \to \sigma_2) \defeq \Rel^-(\sigma_1) \to \Rel^+(\sigma_2).
\]

In the HoRS-to-HoCHC encoding below, we annotate variables with superscripts of not merely sorts but of \emph{interpreted} sorts -- $\hmng{\sigma}$ or $\mmng{\sigma}$ for sort $\sigma$ -- to distinguish HoRS and HoCHC variables.

Let the metavariable $\$$ range over $\nonterms \cup \Sigma \cup V_{\text{RS}}$. We define a transformation $\$ \mapsto \$'$ according to:
\[
\begin{array}{l|l|l}
& \multicolumn{1}{c|}{\$} & \multicolumn{1}{c}{\$'}\\
\hline
\hbox{Variables}  & x : \iota & D_{x'} : \Rel^+(\iota)\\
& x : \sigma_1 \to \sigma_2 & x' : \Rel^+(\sigma_1 \to \sigma_2)\\
\hbox{Terminals} & f : \iota^{n} \to \iota & D_f : \Rel^+(\iota^{n} \to \iota)\\
\hbox{Nonterminals} & F : \sigma & R_F : \Rel^+(\sigma)
\end{array}
\]
where
\[ D_{f} := \lambda x_1 \dots x_{\arity{f}} r. \, (f\,x_1\dots x_{\arity{f}} = r). \]
Note that $R_F$ is a relational variable, but $D_f$ is merely a shorthand; 
$D_f$ is not a symbol -- and neither is $D_{x'}$.
This shorthand allows us to present the relational lift $\lift{-}$ in a simpler way.
Whenever $D_f$ occurs in some encoded HoRS term, it occurs in a fully applied term $D_f\,t_1\dots t_{\arity{f}}\,r$, which is $\beta$-equivalent to $f\,t_1\dots t_{\arity{f}} = r$.
It is this latter term we use in practice (similarly for $D_{x'}$).

For each $F:\sigma \in \nonterms$, we define $\lift{\calR(F)}: \Rel^+(\sigma)$, called the \emph{relational lift}, as follows.
We write $x' : \Rel^-(\tau)$ for the HoCHC variable that is the relational clone of HoRS variable $x : \tau \in \rsvars$, such that:
\[
\lift{\lambda x_1^{\hmng{\sigma_1}} \ldots x_m^{\hmng{\sigma_m}} . \, e} \defeq
\lambda {x'}_1^{\mmng{\Rel^-(\sigma_1)}} \dots {x'}_{m}^{\mmng{\Rel^-(\sigma_m)}}.\, \lift{e}
\]

For HoRS term $\$ \, e_1 \dots e_l$, we define the relational lift as:
\[
\lift{\$ \, e_1 \dots e_l} := 
\begin{array}{l}
\lambda y_1^{\mmng{\Rel^-(\sigma_1)}} \dots y_n^{\mmng{\Rel^-(\sigma_n)}}\, r.\, \\
\qquad \qquad \exists r_1 \dots r_l.  \left(
\begin{array}{ll}
 & \$' \, \dlift{(e_1,r_1)}\dots \dlift{(e_l,r_l)}\,  y_1 \dots y_n\,  r  \\
\land & \bigwedge_{i=1}^l \Prop(e_i,r_i)
\end{array}
\right)
\end{array}
\]
for fresh HoCHC variables $y_1, \dots, y_n$, where $\$: \tau_1 \to \dots{} \to \tau_l \to \sigma_1 \to \dots{} \to \sigma_n \to \iota$ and
\[
    \dlift{(e:\sigma,r)} := \left\{ 
    \begin{array}{ll}
    r & \text{if }\sigma=\iota\\
    \lift{e:\sigma} & \text{o/w}
    \end{array}
    \right.
    \qquad
    \Prop(e:\sigma,r) := \left\{ \begin{array}{ll}
    \lift{e:\sigma}\,r & \text{if }\sigma=\iota\\
    \mathsf{true} & \text{o/w.}
    \end{array}
    \right.
\]

It is worth pointing out that $\lift{e : \sigma}: \Rel^+(\sigma)$, $\dlift{(e:\sigma,r)}:\Rel^-(\sigma)$, and $\Prop(e:\sigma,r):o$.

\subsection{Correctness}
\label{sec:correctness}

Our HoRS-to-HoCHC encoding $\vdash P_\calG : \Delta_\calG$ contains a relational variable for each nonterminal symbol in the original HoRS $\calG = \abra{\calN,\Sigma,\calR,S}$.
We claim that the HoCHC rational variable $R_S: \iota \to o$ corresponding to start symbol $S$ valuates to the characteristic function of $\sem{\calG}$ in the greatest model of $\vdash P_\calG: \Delta_\calG$:

\begin{samepage}
\begin{restatable*}[Correctness]{theorem}{correctnessEquality}
\label{thm:correctness_equality}
\label{thm:homc_correctness}
$\mmng{\Delta_\calG \vdash R_S}(\gfp(T^\calM_{P_\calG:\Delta_\calG})) \, t = 1$ 
if and only if
$t = \sem{\calG}$.
\end{restatable*}
To prove the theorem, we establish a lockstep between iterations of the HoRS endofunction $\hmng{\calG}_\calN$ (in the ascending Kleene chain) and (descending) iterations of the HoCHC one-step consequence operator $T_{\Delta_\calG:P_\calG}^\calM$.
The proof consists of four parts, corresponding to the respective sections of Appendix~\ref{apx:correctness}.

\nopagebreak

First, we define two families of mappings between HoRS semantics and (coinductive) HoCHC semantics.
These mappings allow us to embed HoRS semantics into HoCHC relations.
Second, we show that there exists a $\bot$-free tree $t$ such that $\mmng{\Delta_\calG \vdash \lift{S}}(T_{P_\calG:\Delta_\calG}^{\calM \, n}(\top_{\Delta_\calG}))\,t = 1$, for every iteration $n$ of the one-step consequence operator (Lemma~\ref{lem:correctness_nonemptiness}, ``nonemptiness'').
Third, we show that each $\mmng{\Delta_\calG \vdash \lift{S}}(T_{P_\calG:\Delta_\calG}^{\calM \, n}(\top_{\Delta_\calG}))$ is included in the embedding of $\hmng{\calN \vdash S}(\hmng{\calG}_{\calN}^n(\bot_\calN))$ into HoCHC (Corollary~\ref{cor:correctness_inclusion}, ``inclusion'').
\end{samepage}

Finally, we prove that these ``nonemptiness'' and ``inclusion'' results suffice to show that $R_S$ valuates to the characteristic function of $\sem{\calG}$ in the greatest model of $\vdash P_\calG : \Delta_\calG$ (Theorem~\ref{thm:homc_correctness}).

Appendix~\ref{apx:correctness} details the full proof. 
Although our results in this section pertain to ground sort $\iota$, we require logical relations and proofs lifted to higher-sorts to attain them.
This starts with the families of mappings $\imap_\sigma,\jmap_\sigma$ between HoRS semantics and HoCHC relations that are markedly simpler for sort $\iota$ (Definition~\ref{def:mappings_iota}) than the full mappings defined in Definition~\ref{def:rel_emb}.

\begin{definition}[Embedding of trees into HoCHC relations]
\label{def:mappings_iota}
We define a function $\imap_\iota : \hmng{\iota} \to  \mmng{\Rel^+(\iota)}$ by $\imap_\iota(t) := \lambda s.\,t \sqsubseteq s$, for all $t \in \hmng{\iota}$.
\end{definition}
Note that the function $\imap_\iota$ allows us to embed HoRS trees into HoCHC relations; it is antitone and injective.
Because we are trying to relate a least fixpoint (HoRS semantics) to a greatest fixpoint (coinductive HoCHC), the following lemma is key.
Please refer to Appendix~\ref{apx:correctness} for the proofs.

\begin{restatable}{lemma}{auxiliaryLubGlb}
\label{lem:homc_lub_glb}
For all directed sets $D \subseteq \hmng{\iota}$,
    $\imap_\iota \left(\lub D\right) 
        = \glb \makeset{\imap_\iota(d) \mid d \in D}$.
\end{restatable}

Greatest upper bounds of chains are preserved by the semantics of relationally lifted HoRS terms.

\begin{restatable}{lemma}{auxiliaryDescendingChain}
\label{lem:homc_decreasing_chain_valuations}
For all typing judgements 
\(
\calN \vdash e : \sigma
\) 
of the HoRS $\calG$,
and non-increasing chains of valuations $\calI \subseteq \mmng{\Delta_\calG}$,
\[ 
    \mmng{\Delta_\calG \vdash \lift{e: \sigma}}\left(\glb \calI\right)
        = \glb_{I \in \calI} \mmng{\Delta_\calG \vdash \lift{e: \sigma}}(I).
\]
\end{restatable}

Let us write $\alpha^n$ for $\hmng{\calG}_{\calN}^n(\bot_\calN)$ and $\beta^n$ for $T_{P_\calG:\Delta_\calG}^{\calM \, n}(\top_{\Delta_\calG})$, so that the following hold.

\begin{restatable}[Nonemptiness]{lemma}{correctnessNonemptiness}
\label{lem:correctness_nonemptiness}
\label{lem:homc_nonempty_limit}
There exists a $\bot$-free tree $t \in \mmng{\iota}$ such that $\mmng{\Delta_\calG \vdash \lift{S}}\left(\glb \beta^n\right)\,t$.
\end{restatable}

\begin{restatable}[Inclusion]{corollary}{correctnessInclusion} 
\label{cor:correctness_inclusion}
For all $n \geq 0$,
$\mmng{\Delta_\calG \vdash \lift{S}}(\beta^{n}) \sqsubseteq  \imap_\iota(\hmng{\calN \vdash S}(\alpha^{n}))$.
\end{restatable}

\correctnessEquality
\noindent\emph{Proof.}
\vspace*{-10pt}
\begin{align*}
    \mmng{\Delta_\calG \vdash R_S}(\gfp(T^\calM_{P_\calG:\Delta_\calG}))
                        &= \glb \makeset{ \mmng{\Delta_\calG \vdash \lift{S}}(\beta^n) \mid n \geq 0} \qquad \text{Lem~\ref{lem:homc_decreasing_chain_valuations}}\\
        &\sqsubseteq \glb \makeset{ \imap_\iota(\hmng{\calN \vdash S}(\alpha^n)) \mid n \geq 0} \qquad \text{Cor~\ref{cor:correctness_inclusion}}\\
        &= \imap_\iota\left(\lub \makeset{\hmng{\calN \vdash S}(\alpha^n) \mid n \geq 0}\right) \qquad \text{Lem~\ref{lem:homc_lub_glb}}\\
        &= \imap_\iota\left(\hmng{\calN \vdash S}\left(\lub\alpha^n\right)\right) \qquad \text{Lem~\ref{lem:H_continuous}}\\
                &= \lambda r.\,(\sem{\calG}=r) \qquad \text{$\bot$-freeness}
\end{align*}
Either $\mmng{\Delta_\calG \vdash R_S}(\gfp(T^\calM_{P_\calG:\Delta_\calG}))$ is the constant false function, or it is $\lambda r.\,(\sem{\calG}=r)$.
By Lemma~\ref{lem:homc_nonempty_limit},  $\mmng{\Delta_\calG \vdash R_S}(\gfp(T^\calM_{P_\calG:\Delta_\calG}))$ is not $\lambda r.\,0$, so we conclude that it is $\lambda r.\,(\sem{\calG}=r)$, instead.

It follows that $\mmng{\Delta_\calG \vdash R_S}(\gfp(T^\calM_{P_\calG:\Delta_\calG}))\,t = 1$ if and only if $t = \sem{\calG}$. 
\hfill $\square$


\section{HoRS equivalence problem}
\label{sec:HoRS_equivalence}

The higher-order recursion scheme (HoRS) equivalence problem asks 
whether two given deterministic recursion schemes $\calG_1,\calG_2$ 
generate the same tree 
(i.e.~whether $\sem{\calG_1}=\sem{\calG_2}$, see e.g.~\cite{Ong2015}). 

\nopagebreak

Here, we reduce the HoRS equivalence problem to coinductive HoCHC.
Our procedure formulates a positive and negative instance of the coinductive monotone HoCHC problem over a decidable background theory.
We present the background theory in Section~\ref{sec:homc_maher} and the HoCHC instances in Section~\ref{sec:homc_procedure}.

If coinductive HoCHC is semi-decidable over a decidable background theory -- or our HoRS-to-HoCHC encodings are semi-decidable over Maher's theory of trees in particular~\cite{Maher1988} -- then these two instances can be solved concurrently for a full decision procedure for the HoRS equivalence problem.

\subsection{Maher's theory of trees}\label{sec:homc_maher}

A \emph{theory} $T$ is a set of sentences, 
which is \emph{complete} if 
either $T \vDash \varphi$ or $T \vDash \neg \varphi$, for every sentence $\varphi$.
An \emph{axiomatisation} of an algebra $\mathfrak{A}$ is a recursive set of sentences which are true of $\mathfrak{A}$.
The \emph{theory of an algebra} $\mathfrak{A}$ is a the set of all sentences true of $\mathfrak{A}$.

Maher's (equational) theory of trees $T_\Sigma$ is complete for any finite or infinite alphabet $\Sigma$~\cite{Maher1988}.
It is axiomatised by three axioms:
\begin{align}
    \forall f \in \Sigma. \qquad&\forall \vv{x}\,\vv{y}.\, f\,\vv{x} = f\,\vv{y} \leftrightarrow \vv{x} = \vv{y} \\
    \forall f,g \in \Sigma. \qquad&  \, f \not\equiv g \rightarrow \forall \vv{x}\,\vv{y}.\, f\,\vv{x} \neq g\,\vv{y} \\
    \forall y. \exists ! x. \qquad& x = t(x,y)
\end{align}
where $x=t(x,y)$ ranges over rational solved forms~(see \cite{Maher1988}, p. 355).

In case $\Sigma$ is finite, we need to add the Domain Closure Axiom to obtain completeness:
\begin{align}
    \forall x.\,\bigvee_{f \in \Sigma} \exists \vv{z}.\, x = f\,\vv{z} \tag{DCA}
\end{align}

Fix a ranked alphabet $\Sigma$, viewed as tree constructors.
We write $T_\Sigma$ for the first-order theory of equations of finite and infinite trees constructed from $\Sigma$.
The theory $T_\Sigma$ is complete and decidable, making it an exceedingly appropriate choice of background theory for HoCHC. 
The theory has several models, including the set of finite and infinite trees over $\Sigma$ that we are interested in.

Djelloul et al.~have presented a full first-order constraint solver for (an augmented version of) the theory $T_\Sigma$~\cite{Djelloul2008}.
Questions of expressivity and complexity of the Maher theory are explored in \cite{Colmerauer2003}.
Recent work by Zaiser and Ong has improved the performance of Djelloul et al.'s solver and adapted the theory to algebraic (co)datatypes~\cite{Zaiser2020}.

We are interested in the theory $T_\Sigbot$ over finite alphabet $\Sigbot$, for input HoRS $\calG_1 = \abra{\calN_1,\Sigma,\calR_1,S_1}$ and $\calG_2 = \abra{\calN_2,\Sigma,\calR_2,S_2}$. 
Note that the assumption that both HoRS have the same alphabet $\Sigma$ is WLOG; if they have distinct alphabets, we can take $\Sigma$ to be their union.
In the Maher theory $T_\Sigbot$, the ``unfinished'' tree $\bot$ is treated as any other nullary terminal symbol.

\subsection{Decision procedure}\label{sec:homc_procedure}

Let $\calG_1 = \langle \mathcal{N}_1, \Sigma, \mathcal{R}_1, S_1\rangle$ and $\calG_2 = \langle \mathcal{N}_2, \Sigma, \mathcal{R}_2, S_2\rangle$ be deterministic HoRS.
Assume the trees they generate are $\bot$-free, which is WLOG due to Section~\ref{sec:HoRS_botfree}.
Consider these HoCHC goal formulas:
\[
\mathit{Eq}_1 \defeq \exists r_1\, r_2 \ldotp (R_{S_1} \, r_1 \wedge R_{S_2} \, r_2) \wedge (r_1 = r_2)
\qquad\quad
\mathit{Eq}_0 \defeq \exists r_1\, r_2 \ldotp (R_{S_1} \, r_1 \wedge R_{S_2} \, r_2) \wedge (r_1 \not= r_2)
\]
Using the definitions from Section~\ref{sec:HoRS_transformation}, we define HoCHC problems 
$\mathcal{P}_i := \abra{\Delta_{\calG_1} \cup \Delta_{\calG_2}, P_{\calG_1} \cup P_{\calG_2}, Eq_i}$,
for $i \in \makeset{0,1}$, with the Maher theory $T_\Sigbot$ as the constraint language and the set $\alltrees$ of finite and infinite trees as the designated model.

\begin{samepage}
Thanks to Theorem~\ref{thm:homc_correctness}, we have: $\sem{\calG_1} = \sem{\calG_2}$ iff $\mathcal{P}_1$ is solvable, and $\sem{\calG_1} \neq \sem{\calG_2}$ iff $\mathcal{P}_0$ is solvable.

\nopagebreak
Recall that the Maher theory $T_\Sigbot$ is decidable -- to be exact, the question $T_\Sigbot \models \phi$ for first-order tree constraints $\phi$ like $r_1 = r_2$ and $r_1 \neq r_2$ above.\end{samepage}
Note, however, that $\mathcal{P}_1$ and $\mathcal{P}_0$ are coinductive HoCHC problems. 
It is an open question whether coinductive HoCHC problems over a (semi-)decidable background theory -- like the Maher theory $T_\Sigbot$ -- can be semi-decided via a reduction to a first-order problem, like inductive HoCHC can~\cite{Pham2018,OngWagner2019}.
If there exists a such semi-decision procedure for solving (monotone) coinductive HoCHC over $T_\Sigbot$, then we can decide $\sem{\calG_1} = \sem{\calG_2}$ by dovetailing our two HoCHC problems.

The full ``decision'' procedure for the HoRS equivalence problem is outlined in Figure~\ref{fig:decision_procedure}, so that:
\begin{restatable}{theorem}{HoRSEquivalenceDecidability}
\label{thm:HoRS_equivalence}
The HoRS equivalence problem is decidable if the HoRS-to-HoCHC encoding lives in a semi-decidable fragment of coinductive HoCHC over Maher's complete and decidable theory of trees.
\end{restatable}


\section{Conclusion and related work}
\label{sec:conclusion}

\paragraph{Higher-order recursion scheme equivalence problem.}
To the best of our knowledge, the HoRS equivalence problem~\cite{Clairambault2013} remains open.
We obtain a full decision procedure for this problem if:
\begin{enumerate*}[label=(\alph*)]
    \item coinductive HoCHC over Maher's theory of trees~\cite{Maher1988} is semi-decidable, or 
    \item the image of our HoRS-to-HoCHC encoding over Maher's theory lives in a semi-decidable fragment of coinductive HoCHC.
\end{enumerate*}

Restricted to order 1, the HoRS equivalence problem is equivalent to the DPDA equivalence problem \cite{Courcelle1978}.
Thus, the fundamental result of \cite{DBLP:journals/tcs/Senizergues01}, and subsequent refinements by \cite{DBLP:journals/tcs/Stirling01}, \cite{DBLP:journals/tcs/Senizergues02}, and \cite{DBLP:conf/lics/Jancar12} provide a decision procedure for the equivalence of first-order HoRS.

\paragraph{$\lambda\textbf{Y}$-calculus B\"ohm tree equivalence problem.}
The HoRS equivalence problem is recursively equivalent to $\lambda{\bf Y}$-calculus B\"{o}hm tree equivalence problem, 
which asks whether the B\"{o}hm trees of two given $\lambda{\bf Y}$-terms are equal \cite{Clairambault2013}.
The question of the decidability ``has been there from the beginning of the subject'' \cite{Walukiewicz16}.
Semi-decidability of coinductive HoCHC would also allow us to decide this problem.

Note that the closely related $\lambda\mathbf{Y}$-calculus word problem (are two closed $\lambda{\bf Y}$-terms $\beta\eta{\bf Y}$-equivalent?) is undecidable~\cite{Statman2004}.
Although HoRS are programs of a simply-typed $\lambda{\bf Y}$-calculus, constructed from uninterpreted function symbols,
they define a strict subsystem of the $\lambda{\bf Y}$-calculus: the same set of trees as ground-type $\lambda{\bf Y}$-terms with free variables (corresponding to terminal symbols) of order at most 1~\cite{Salvati2014}. 

\paragraph{Semi-decidability of coinductive HoCHC.}

Existing semi-decidability results for inductive HoCHC \cite{Pham2018,OngWagner2019} do not carry over to coinductive HoCHC, because proofs for coinductive programs may have infinite length.

For first-order Horn clauses, Coinductive Logic Programming (CoLP, \cite{Gupta2007,Simon2007}) provides an approach to computing solutions for infinite sequences of reductions.
Resolution proof systems for coinductive logic programs rely on loop detection in infinite proofs, see e.g.~\cite{Komendantskaya2013,Komendantskaya2018} and refinements~\cite{Komendantskaya2017,Basold2019}. 
Intuitively, our characterisation of HoRS in HoCHC has not necessarily made such loop detection computationally simpler.

Our best hope is that our HoRS-to-HoCHC encodings live in a semi-decidable fragment of coinductive HoCHC.
There are some indications that this could be the case, e.g.~Lemma~\ref{lem:homc_decreasing_chain_valuations} shows the semantics of encoded HoRS behaves better than (monotone) coinductive HoCHC as a whole.

\paragraph{Relation to HFL.}
In recent years, HFL model checking -- where properties are expressed in higher-order modal fixpoint logic~\cite{Viswanathan2004} -- has gained traction~\cite{Kobayashi2017,Kobayashi2018}.
HoCHC roughly corresponds to a fragment of HFL$_\mathbb{Z}$ without modal operators and fixpoint alternations. HoCHC unsolvability captures HFL \emph{non-reachability}~\cite{Kobayashi2018}.
It seems that coinductive HoCHC unsolvability corresponds to \emph{must-reachability}. Clarifying this relation may help us understand the complexity of coinductive HoCHC.

\bibliographystyle{eptcs}
\bibliography{refs}

\clearpage

\appendix 


\section{Correctness proofs from Section~\ref{sec:transformation}}
\label{apx:correctness}

\subsection{Mappings between HoRS and HoCHC semantics}\label{sec:homc_corr_mappings}

We define a relaxation of the monotone HoCHC sort frame that we call \emph{relatively monotone}.
This new sort frame coincides with our trusted monotone sort frame for sorts $\iota$ and $\iota \to o$.

\begin{definition}[Relatively monotone sort frame]
For each sort $\sigma$ over $\iota$, we define
\[
    \Image_\sigma \defeq \makeset{\theta \in \dmng{\Rel^-(\sigma)} \mid \exists h \in \hmng{\sigma}.\,\imap^-_\sigma(h) = \theta}
\]
where $\imap^-_\sigma$ is defined as in Definition~\ref{def:rel_emb}, and $\dmng{-}$ denotes the \emph{relatively monotone frame}:
\begin{align*}
    \dmng{\iota} &\defeq \mmng{\iota}\\
    \dmng{\Rel^+(\iota)} &\defeq \mmng{\iota \to o}\\
    \dmng{\Rel^+(\sigma_1 \to \sigma_2)} &\defeq \left[ \dmng{\Rel^-(\sigma_1)} \Rightarrow_{m[\Image_{\sigma_1}]} \dmng{\Rel^+(\sigma_2)} \right]
\end{align*}
The latter denotes the space of functions that are \emph{monotone with respect to $\Image_{\sigma_1}$}, i.e.~$f: \dmng{\Rel^-(\sigma_1)} \to \dmng{\Rel^+(\sigma_2)}$ is an element of $\dmng{\Rel^+(\sigma_1 \to \sigma_2)}$ just if: $z_1 \sqsubseteq z_2$ implies $f\,z_1 \sqsubseteq f\,z_2$ for all $z_1,z_2 \in \Image_{\sigma_1}$.
\end{definition}

For relational sorts $\rho$ larger than $\iota \to o$, $\dmng{\rho}$ captures a strictly larger set of functions than $\mmng{\rho}$.
The following definition extends Definition~\ref{def:mappings_iota} to higher sorts. 
\begin{definition}\label{def:rel_emb}
For all sorts $\sigma$ over $\iota$, we define two pairs of mappings:
\[
    \dmng{\Rel^+(\sigma)} \galois{\jmap_\sigma}{\imap_\sigma} \hmng{\sigma}
    \qquad 
    \dmng{\Rel^-(\sigma)} \galois{\jmap^-_\sigma}{\imap^-_\sigma} \hmng{\sigma}
\]
For sort $\iota$, all $t \in \hmng{\iota}$ and $p \in \dmng{\Rel^+(\iota)}$, we define:
\[ 
\imap_\iota(t) := \lambda s.\,t \sqsubseteq s
\qquad 
\jmap_\iota(p) := 
\left\{\begin{array}{ll}
\bot & \quad \text{if }p = \lambda s.\,0\\
\mathit{choice} \, \min \makeset{t \mid p\,t} & \quad \text{otherwise}
\end{array}\right.
\]
where {\normalfont choice} denotes an arbitrary choice function, which exists by the Axiom of Choice.

For $\sigma = \sigma_1 \to \dots \to \sigma_m \to \iota$ with $m > 0$, we define the following for all $h \in \hmng{\sigma}$:
\[
\imap_\sigma \, h := \lambda x_1^{\dmng{\Rel^-(\sigma_1)}} \dots x_m^{\dmng{\Rel^-(\sigma_m)}}.\,\imap_\iota \left( h \, (\jmap^-_{\sigma_1}\, x_1) \dots (\jmap^-_{\sigma_m}\, x_m) \right)
\]
with
\[
\jmap^-_\sigma : \dmng{\Rel^-(\sigma)} \to \hmng{\sigma} := 
\left\{
\begin{array}{ll}
\text{inclusion }\dmng{\iota} \hookrightarrow \hmng{\iota} & \quad \text{if }\sigma = \iota\\
\jmap_\sigma & \quad \text{otherwise}.
\end{array}\right.
\]
Similarly, for all $\theta \in \dmng{\Rel^+(\sigma)}$:
\[
\jmap_{\sigma}\,\theta := 
\mathit{choice} \, \max \makeset{h \in \hmng{\sigma} \mid \theta \sqsubseteq \imap_\sigma\,h}
\]
with
\[
\imap^-_\sigma : \hmng{\sigma} \to \dmng{\Rel^-(\sigma)} := \left\{
\begin{array}{ll}
\text{inclusion }\hmng{\iota} \hookrightarrow \dmng{\iota} & \quad \text{if }\sigma = \iota\\
\imap_\sigma & \quad \text{otherwise}.
\end{array}\right.
\]
\end{definition}

The top case of $\jmap_\iota$ will not used in practice.

\begin{lemma}\label{lem:homc_wellsorted_pair}
For all sorts $\sigma$ over $\iota$,
\vspace*{-7pt}
\begin{multicols}{2}
\begin{enumerate}[label=(\arabic*)]
    \item $\imap_\sigma \left(\bot_{\hmng{\sigma}}\right) = \top_{\dmng{\Rel^+(\sigma)}}$ \label{enum:homc_imap_extremes}
            \item $\imap_\sigma$ is injective \label{enum:homc_imap_injective}
    \item $\imap_\sigma$ is antitone \label{enum:homc_imap_antitone}
    \item $\jmap_\sigma \circ \imap_\sigma = \id_{\hmng{\sigma}}$ \label{enum:homc_imap_id}
    \item $\jmap^-_\sigma \circ \imap^-_\sigma = \id_{\hmng{\sigma}}$ \label{enum:homc_imap_id_minus}
    
    \vspace*{15pt}
\end{enumerate}
\end{multicols}
\end{lemma}

\auxiliaryLubGlb*
\begin{proof}
Recall that $\hmng{\iota}$ is a dcpo and $\mmng{\Rel^+(\iota)}$ a complete lattice, so that the bounds are defined.
\[
    \imap_\iota \left(\lub D\right)
        = \lambda s.\left(\lub D \sqsubseteq s\right)
        = \lambda s.\,\glb\makeset{d \sqsubseteq s \mid d \in D} 
        = \glb\makeset{\lambda s.\,  (d \sqsubseteq s) \mid d \in D}
        = \glb\makeset{\imap_\iota ( d ) \mid d \in D} 
\]
To see that the second equality holds: 
suppose that $\lub D \sqsubseteq s$ for some $s \in \dmng{\iota} = \hmng{\iota}$.
By transitivity, $d \sqsubseteq \lub D \sqsubseteq s$ for all $d \in D$.
Since the greatest lower bound on $\mmng{o}$ is conjunction, this implies that
$\glb \makeset{d \sqsubseteq s \mid d \in D}$.
For the converse, suppose $\glb \makeset{d \sqsubseteq s \mid d \in D}$. 
This means that $d \sqsubseteq s$, for all $d \in D$.
Thus, $s$ is an upper bound on $D$.
However, $\lub D$ is the least upper bound on this set, so $\lub D \sqsubseteq s$.
\end{proof}

\subsection{Nonemptiness}\label{sec:homc_corr_nonempty}

We aim to show there exists a tree $t$ that does not contain $\bot$ such that $\mmng{\Delta_\calG \vdash \lift{S}}(T_{P_\calG:\Delta_\calG}^{\calM \, n}(\top_{\Delta_\calG}))\,t = 1$, for every $n \geq 0$ (Corollary~\ref{cor:homc_cor_nonempty}).
To this end, we define a family of logical relations in Definition~\ref{def:homc_botfree_relations} to capture this notion at higher sorts and for larger sort environments, as proved in Lemma~\ref{lem:homc_botfree_lemma}.

Intuitively, such a relation holds whenever a predicate maps nonempty inputs to nonempty outputs, where ``nonempty'' is taken to mean with respect to $\bot$-free trees.

\begin{definition}\label{def:homc_botfree_relations}
We define a family of logical relations $\bot\mhyphen\mathsf{free}_\sigma \subseteq \mmng{\sigma}$:
\begin{align*}
    \bot\mhyphen\mathsf{free}_\iota(t) &\defeq t\text{ has no }\bot\text{-labelled leaves}\\
    \bot\mhyphen\mathsf{free}_{\Rel^+(\iota)}(p) &\defeq \exists s \in \mmng{\iota}.\, \bot\mhyphen\mathsf{free}_\iota(s) \land p\,s = 1\\
    \bot\mhyphen\mathsf{free}_{\Rel^+(\sigma_1\to\sigma_2)}(p) &\defeq \forall s \in \mmng{\Rel^-(\sigma_1)}.\, \bot\mhyphen\mathsf{free}_{\Rel^-(\sigma_1)}(s) \Rightarrow \bot\mhyphen\mathsf{free}_{\Rel^+(\sigma_2)}(p\,s)\\
    \bot\mhyphen\mathsf{free}_{\Rel^-(\Gamma)}(\theta)
        &\defeq \dom(\lift{\Gamma}) = \dom(\theta) \land \bigwedge_{x':\Rel^-(\sigma) \in \lift{\Gamma}} \bot\mhyphen\mathsf{free}_{\Rel^-(\sigma)}(\theta(x'))
\end{align*}
Alternatively, we can define, for all $\vv{s}$ from the appropriate domains:
\vspace*{-5pt}
\[
    \bot\mhyphen\mathsf{free}_{\Rel^+(\sigma_1\to\dots\to \sigma_m \to \iota)}(p)
    \defeq \forall \vv{s}. \left(\bigwedge_{i \in [m]} \bot\mhyphen\mathsf{free}_{\Rel^-(\sigma_i)}(s_i)\right) \Rightarrow \bot\mhyphen\mathsf{free}_{\Rel^+(\iota)}(p\,s_1 \dots s_m)
\]
\end{definition}

\begin{lemma}
For all $p_1,p_2 \in \mmng{\Rel^+(\iota)}$,
if $p_1 \sqsubseteq p_2$ and $\bot\mhyphen\mathsf{free}_{\Rel^+(\iota)}(p_1)$,
then $\bot\mhyphen\mathsf{free}_{\Rel^+(\iota)}(p_2)$.
\end{lemma}
\begin{proof}
Trivial, using that witness $t\in \mmng{\iota}$ of $\bot\mhyphen\mathsf{free}_{\Rel^+(\iota)}(p_1)$ also witnesses $\bot\mhyphen\mathsf{free}_{\Rel^+(\iota)}(p_2)$.
\end{proof}

\begin{lemma}\label{lem:homc_botfree_lemma}
For all $n \geq 0$, 
all typing judgements 
\(
\calN, \Gamma \vdash e : \sigma
\) of the HoRS $\calG$ where $\Gamma = \set{x_1 : \tau_1, \ldots, x_k : \tau_k}$,
and valuations $\theta \in \mmng{\lift{\Gamma}}$,
\[
    \calN, \Gamma \vdash e:\sigma \land \bot\mhyphen\mathsf{free}_{\Rel^-(\Gamma)}(\theta) 
        \quad \Rightarrow \quad
    \bot\mhyphen\mathsf{free}_{\Rel^+(\sigma)}\left(\mmng{\Delta_\calG, \lift{\Gamma} \vdash \lift{e}}(\beta^{n})\right)
\]
where 
\begin{align*}
\lift{\Gamma} &:= \set{x'_1: \Rel^-(\tau_1), \dots, x'_k: \Rel^-(\tau_k)}\\
\beta^{n} &:= \left( T_{P_\calG:\Delta_\calG}^{\calM \, n}(\top_{\Delta_\calG})\right) [\overline{x'} \mapsto \vv{\theta(x')}] \in \mmng{\Delta_\calG, \lift{\Gamma}}.
\end{align*}
Notice that $\calN, \Gamma \vdash e : \sigma$ implies $\Delta_\calG, \lift{\Gamma} \vdash \lift{e} : \Rel^+(\sigma)$.
\end{lemma}
\begin{proof}
We proceed by induction on $n \geq 0$ within which (both in the base case and the induction step) we use structural induction on HoRS term $e$.
Some parts of the proof are presented out of order to avoid duplication.
Figure~\ref{fig:homc_botfree_IH_overview} outlines the structure of the proof.
We assume WLOG that $e$ contains no $\lambda$s.

We use the following shorthand, for $e:\sigma = \sigma_1 \to \dots \to \sigma_m \to \iota$ and $z_i \in \mmng{\Rel^-(\sigma_i)}$ such that $\bot\mhyphen\mathsf{free}_{\Rel^-(\sigma_i)}(z_i)$, for all $i \in [m]$:
\phantomsection\label{eq:homc_A_botfree}
\begin{align*}
    A &= \mmng{\Delta_\calG, \lift{\Gamma} \vdash \lift{e:\sigma}}(\beta^{n})\,\vv{z}
\end{align*}
Our proof strategy is to rewrite $\eqA$ and provide a witness to $\bot\mhyphen\mathsf{free}_{\Rel^+(\iota)}(\mmng{\Delta_\HORS, \lift{\Gamma} \vdash \lift{e:\sigma}}(\beta^{n})\,\vv{z})$, which proves $\bot\mhyphen\mathsf{free}_{\Rel^+(\sigma)}(\mmng{\Delta_\HORS, \lift{\Gamma} \vdash \lift{e:\sigma}}(\beta^{n}))$.

We present the following base cases w.r.t. the structure of $e$ (also denoted $b$ for \emph{base case expression}).

\emph{Case $e = x: \iota \in \rsvars$.}
For all $n \geq 0$, the tree $\theta(x')$ is a witness thanks to $\bot\mhyphen\mathsf{free}_{\Rel^-(\Gamma)}(\theta)$:
\[
    \eqA
        = \mmng{\Delta_\HORS, \lift{\Gamma} \vdash \lift{x}}(\beta^{n})\\
        = \mmng{\Delta_\HORS, \lift{\Gamma} \vdash \lambda r.\,x' = r}(\beta^{n})\\
                        = \lambda s. \left(\beta^{n}(x') = s\right)\\
        = \lambda s. \left(\theta(x') = s\right)
\]

\emph{Case $e = f: \sigma_1 \to \dots \to \sigma_m \to \iota \in \Sigma$.}
For all $n \geq 0$, the $\bot$-free tree $\widehat{F_f} \, \vv{z}$ is a witness:
\[
    \eqA
        = \mmng{\Delta_\HORS, \lift{\Gamma} \vdash \lift{f}}(\beta^{n})\,\vv{z}\\
        = \mmng{\Delta_\HORS, \lift{\Gamma} \vdash \lambda \vv{y} \, r.\,f\,\vv{y} = r}(\beta^{n})\,\vv{z}\\
        = \lambda s . \left(\widehat{F_f} \, \vv{z} = s \right)
\]

\emph{Case $e = x: \sigma_1 \to \dots \to \sigma_m \to \iota \in \rsvars$ and $m>0$.}
For all $n \geq 0$, it holds that $\bot\mhyphen\mathsf{free}_{\Rel^+(\iota)}(\theta(x')\,\vv{z})$:
\[
    \eqA
        = \mmng{\Delta_\HORS, \lift{\Gamma} \vdash \lift{x}}(\beta^{n})\,  \vv{z}\\
        = \mmng{\Delta_\HORS, \lift{\Gamma} \vdash \lambda \vv{y}\,r.\,x'\,\vv{y}\,r}(\beta^{n})\,  \vv{z}\\
        = \lambda s.\,\beta^{n} (x')\, \vv{z} \, s\\
                = \theta(x')\, \vv{z}
\]

\emph{Case $e = F:\sigma_1 \to \dots \to \sigma_m \to \iota \in \calN$ and $n = 0$.}
Any $\bot$-free tree (e.g.~$\sem{\calG}$) is a witness:
\[
    \eqA 
        = \mmng{\Delta_\HORS, \lift{\Gamma} \vdash \lift{F}}(\beta^{0}) \, \vv{z}\\
        = \mmng{\Delta_\HORS, \lift{\Gamma} \vdash \lambda \vv{y} \, r.\,R_F\,\vv{y} \, r}(\beta^{0})\,\vv{z}\\
        = \lambda s.\,\beta^{0}(R_F) \, \vv{z} \, s \\
        = \top_{\mmng{\Rel^+(\sigma)}} \, \vv{z}
\]

This covers $n = 0$ for all base case expressions.
We distinguish three induction hypotheses, where $S(n,e)$ denotes that the claim holds for $n$ and expression $e$.
The proof consists of four parts (in a logical sense but not a physical, to prevent duplication) which are related as in Figure~\ref{fig:homc_botfree_IH_overview}.
Thus, we have now proved $S(0,b)$ for all base case expressions $b$. 
Next, we use~\ref{IH:nonempty_saf_0} to show $S(0,e)$ for all expressions $e$.

\vspace*{-15pt}
\begin{figure}[hbt]
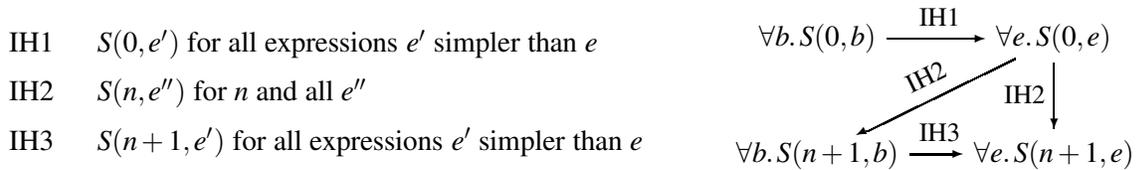

    \centering
    \begin{subfigure}[c]{0.55\textwidth}
        \begin{enumerate}[label=IH\arabic*]
        \item $\quad S(0,e')$ for all expressions $e'$ simpler than $e$ \label{IH:nonempty_saf_0}
        \item $\quad S(n,e'')$ for $n$ and all $e''$ \label{IH:nonempty_saf_BC_n+1}
        \item $\quad S(n+1,e')$ for all expressions $e'$ simpler than $e$ \label{IH:nonempty_saf_IC_n+1}
        \end{enumerate}
    \end{subfigure}
    \begin{subfigure}[c]{0.43\textwidth}
        \begin{diagram}[h=2em,labelstyle=\textstyle,midshaft]
            \forall b.\,S(0,b) & \rTo^{\;\small\text{\ref{IH:nonempty_saf_0}}} & \forall e.\,S(0,e)\\
             & \ldTo^{\small\text{\ref{IH:nonempty_saf_BC_n+1}}} & \dTo^{\small\text{\ref{IH:nonempty_saf_BC_n+1}}} \\
             \forall b.\,S(n+1,b) & \rTo^{\;\small\text{\ref{IH:nonempty_saf_IC_n+1}}} & \forall e.\,S(n+1,e)
        \end{diagram}
        \vspace*{3pt}
    \end{subfigure}
    \vspace*{-12pt}
    \caption{The inductive structure of the correctness proof of the HoRS-to-HoCHC encoding.}
    \label{fig:homc_botfree_IH_overview}
\end{figure}

In this inductive case, we consider expressions $e = \$ \, e_1 \dots e_\ell : \sigma_1 \to \dots \to \sigma_m \to \iota$ for some $\ell > 0$.
As before,
let $z_i \in \mmng{\Rel^-(\sigma_i)}$ such that $\bot\mhyphen\mathsf{free}_{\Rel^-(\sigma_i)}(z_i)$,
for each $i \in [m]$.
We introduce some shorthands:
\[
    \Delta_\HORS' \defeq \Delta_\HORS, \lift{\Gamma}, \vv{y}, r \qquad
    \Delta_\HORS'' \defeq \Delta_\HORS, \lift{\Gamma},\vv{y}, r, \vv{r} \qquad
    \beta^{n,\overline{z},s} \defeq \beta^{n}[\vv{y} \mapsto \vv{z}, r \mapsto s]
\]
\begin{samepage} 
\!Note that the sort $\tau$ of $\$$ is of the form
\[ \tau = \tau_1 \to \dots \to \tau_\ell \to \sigma_1 \to \dots \to \sigma_m \to \iota \]
where $e_1 : \tau_1, \dots, e_\ell : \tau_\ell$, for some $\ell > 0$.
Sometimes we abbreviate $\sigma_1 \to \dots \to \sigma_m \to \iota$ to $\sigma$.
\end{samepage}

\phantomsection\label{eq:homc_dag_nonemptiness_proof}
In the sequel, steps marked with $\dag$ use \ref{IH:nonempty_saf_0} for $n=0$, and \ref{IH:nonempty_saf_IC_n+1} for $n>0$. 

\phantomsection\label{eq:A_rewrite}
For all $n \geq 0$ and expressions $e = \$ \, \vv{e}$, we can rewrite $\eqA$:
\begin{align*}
    \eqA 
        &= \mmng{\Delta_\HORS, \lift{\Gamma} \vdash \lift{\$\,\vv{e}}}(\beta^{n})\, \vv{z}\\
        &= \mmng{\Delta_\HORS, \lift{\Gamma} \vdash \lambda \vv{y}\,r.\, \exists \vv{r}.\, \$' \, \dlift{(e_1,r_1)}\dots \dlift{(e_\ell,r_\ell)}\, \vv{y} \, r \land \bigwedge_{i \in [\ell]} Prop(e_i,r_i)}(\beta^{n})\, \vv{z}\\
        &= \lambda s.\,\mmng{\Delta_\HORS' \vdash \exists \vv{r}.\, \$' \, \dlift{(e_1,r_1)}\dots \dlift{(e_\ell,r_\ell)}\, \vv{y} \, r \land \bigwedge_{i \in [\ell]} Prop(e_i,r_i)}(\beta^{n,\overline{z},s})\\
                                                        &= \lambda s.\,\max\Big\{ \\
            & \qquad \qquad \qquad \min\{ \\
                & \qquad \qquad \qquad \qquad \,\mmng{\Delta_\HORS'' \vdash \$'}(\beta^{n,\overline{z},s}[\vv{r} \mapsto \vv{r'}]) ( \mmng{\Delta_\HORS'' \vdash \dlift{(e_1,r_1)}}(\beta^{n,\overline{z},s}[\vv{r} \mapsto \vv{r'}])) \\
                & \qquad \qquad \qquad \qquad \qquad \qquad \qquad \quad\quad\, \dots (\mmng{\Delta_\HORS'' \vdash \dlift{(e_\ell,r_\ell)}}(\beta^{n,\overline{z},s}[\vv{r} \mapsto \vv{r'}]))\, \vv{z} \, s, \\
                & \qquad \qquad \qquad \qquad \min\{\mmng{\Delta_\HORS'' \vdash Prop(e_i,r_i)}(\beta^{n,\overline{z},s}[\vv{r} \mapsto \vv{r'}]) \mid i \in [\ell]\} \\
            & \qquad \qquad \qquad \;\} \\
        & \qquad \quad  \mid \forall i \in [\ell].\, r'_i \in \mmng{\Rel^-(\tau_i)} \Big\} 
        \end{align*}
We now distinguish two cases for each subexpression $e_i:\tau_i$, namely $\tau_i = \iota$ and $\tau_i = \tau'_1 \to \tau'_2$.

If $e_i$ is of sort $\iota$, then the following holds.
\begin{align*}
    \mmng{\Delta_\HORS'' \vdash Prop(e_i:\iota,r_i)}(\beta^{n,\overline{z},s}[\vv{r} \mapsto \vv{r'}])
                        &= \mmng{\Delta_\HORS,\lift{\Gamma} \vdash \lift{e_i}}(\beta^{n})\,r_i'\\
    \mmng{\Delta_\HORS'' \vdash \dlift{(e_i:\iota,r_i)}}(\beta^{n,\overline{z},s}[\vv{r} \mapsto \vv{r'}])
                &= r_i'
\end{align*}
We know from \hyperref[eq:homc_dag_nonemptiness_proof]{$\dag$} that $\bot\mhyphen\mathsf{free}_{\Rel^+(\iota)}(\mmng{\Delta_\HORS,\lift{\Gamma} \vdash \lift{e_i}}(\beta^{n}))$. This means that there exists $r'' \in \mmng{\iota}$ such that $\bot\mhyphen\mathsf{free}_{\iota}(r'')$ and $\mmng{\Delta_\HORS,\lift{\Gamma} \vdash \lift{e_i}}(\beta^{n})\,r'' = 1$.

Otherwise, in case $e_i:\tau_i = \tau'_1 \to \tau'_2$, the following holds:
\begin{align*}
    \mmng{\Delta_\HORS'' \vdash Prop(e_i:\tau_i,r_i)}(\beta^{n,\overline{z},s}[\vv{r} \mapsto \vv{r'}]) 
                &= 1  \\
    \mmng{\Delta_\HORS'' \vdash \dlift{(e_i:\tau_i,r_i)}}(\beta^{n,\overline{z},s}[\vv{r} \mapsto \vv{r'}]) 
                &= \mmng{\Delta_\HORS, \lift{\Gamma} \vdash \lift{e_i}}(\beta^{n})
\end{align*}
We know from \hyperref[eq:homc_dag_nonemptiness_proof]{$\dag$} that $\bot\mhyphen\mathsf{free}_{\Rel^+(\tau_i)}(\mmng{\Delta_\HORS,\lift{\Gamma} \vdash \lift{e_i}}(\beta^{n}))$.

As ``semantic equivalents'' of the above terms, let us write 
\[
    P_i:o \defeq
        \left\{\begin{array}{ll}
            \mmng{\Delta_\HORS,\lift{\Gamma} \vdash \lift{e_i}}(\beta^{n})\,r_i' & \quad \text{if }\tau_i = \iota \\
            1 & \quad \text{if }\tau_i = \tau'_1 \to \tau'_2
        \end{array}\right.
\]
and
\[
    T_i:\Rel^-(\tau_i) \defeq
        \left\{\begin{array}{ll}
            r_i' & \quad \text{if }\tau_i = \iota \\
            \mmng{\Delta_\HORS, \lift{\Gamma} \vdash \lift{e_i}}(\beta^{n}) & \quad \text{if }\tau_i = \tau'_1 \to \tau'_2
        \end{array}\right.
\]
for all $i \in [\ell]$.
Additionally, we define:
\[
    S_i:\Rel^-(\tau_i) \defeq
        \left\{\begin{array}{ll}
            r_i'' & \quad \text{if }\tau_i = \iota \\
            \mmng{\Delta_\HORS, \lift{\Gamma} \vdash \lift{e_i}}(\beta^{n}) & \quad \text{if }\tau_i = \tau'_1 \to \tau'_2
        \end{array}\right.
\]
where $r_i''$ is an arbitrary ($\bot$-free) witness to $\bot\mhyphen\mathsf{free}_{\Rel^+(\iota)}(\mmng{\Delta_\HORS,\lift{\Gamma} \vdash \lift{e_i}}(\beta^{n}))$, which exists by \hyperref[eq:homc_dag_nonemptiness_proof]{$\dag$}.
This gives us $\bot\mhyphen\mathsf{free}_{\Rel^-(\tau_i)}(S_i)$ for all $i \in [m]$.

We derive by abuse of notation, using the above:
\begin{align*}
    \eqA
        &= \lambda s.\,\exists \vv{r'}. \left( \mmng{\Delta_\HORS'' \vdash \$'}(\beta^{n,\overline{z},s}[\vv{r} \mapsto \vv{r'}]) \, \vv{T}\; \vv{z} \, s \land \bigwedge_{i \in [\ell]} P_i \right)
\end{align*}

We continue by case analysis on $\$$. 

\emph{Case $e = f \, e_1 \dots e_\ell$ with $f \in \Sigma$.}
For all $n \geq 0$, the $\bot$-free tree $\widehat{F_f}\,\vv{r''}\,\vv{z}$ is a witness:
\begin{align*}
    \eqA 
        &= \lambda s.\,\exists \vv{r'}. \left( \mmng{\Delta_\HORS'' \vdash D_f}(\beta^{n,\overline{z},s}[\vv{r} \mapsto \vv{r'}]) \, \vv{r'}\; \vv{z} \, s \land \bigwedge_{i \in [\ell]} \mmng{\Delta_\HORS,\lift{\Gamma} \vdash \lift{e_i}}(\beta^{n})\,r_i' \right) \\
        &= \lambda s.\,\exists \vv{r'}. \left( \widehat{F_f} \, \vv{r'}\; \vv{z} = s \land \bigwedge_{i \in [\ell]} \mmng{\Delta_\HORS,\lift{\Gamma} \vdash \lift{e_i}}(\beta^{n})\,r_i' \right) \\
        &\sqsupseteq \lambda s. \left( \widehat{F_f} \, \vv{r''}\; \vv{z} = s \right) \qquad \hyperref[eq:homc_dag_nonemptiness_proof]{\dag}
\end{align*}

\emph{Case $e = x \, e_1 \dots e_\ell$ with $x \in \rsvars$.}
For all $n \geq 0$, it holds that $\bot\mhyphen\mathsf{free}_{\Rel^+(\iota)}(\theta(x') \, \vv{S}\, \vv{z})$ and:
\begin{align*}
    \eqA 
        &= \lambda s.\,\exists \vv{r'}. \left( \mmng{\Delta_\HORS'' \vdash x'}(\beta^{n,\overline{z},s}[\vv{r} \mapsto \vv{r'}]) \, \vv{T}\; \vv{z} \, s \land \bigwedge_{i \in [\ell]} P_i \right)\\
        &= \lambda s.\,\exists \vv{r'}. \left( \beta^{n,\overline{z},s}[\vv{r} \mapsto \vv{r'}](x') \, \vv{T}\; \vv{z} \, s \land \bigwedge_{i \in [\ell]} P_i \right)\\
                &= \lambda s.\,\exists \vv{r'}. \left( \theta(x') \, \vv{T}\; \vv{z} \, s \land \bigwedge_{i \in [\ell]} P_i \right)\\
        &\sqsupseteq \lambda s. \, \theta(x') \, \vv{S}\; \vv{z} \, s \qquad \hyperref[eq:homc_dag_nonemptiness_proof]{\dag} \\
        &= \theta(x') \, \vv{S}\; \vv{z}  
\end{align*}

\emph{Case $e = F \, e_1 \dots e_\ell$ with $F:\sigma \in \calN$, and $n=0$.}
Any $\bot$-free tree (e.g.~$\sem{\calG}$) is a witness:
\begin{align*}
    \eqA 
        &= \lambda s.\,\exists \vv{r'}. \left( \mmng{\Delta_\HORS'' \vdash R_F}(\beta^{0,\overline{z},s}[\vv{r} \mapsto \vv{r'}]) \, \vv{T}\; \vv{z} \, s \land \bigwedge_{i \in [\ell]} P_i \right)\\
        &= \lambda s.\,\exists \vv{r'}. \left( \beta^{0,\overline{z},s}[\vv{r} \mapsto \vv{r'}](R_F) \, \vv{T}\; \vv{z} \, s \land \bigwedge_{i \in [\ell]} P_i \right)\\
        &= \lambda s.\,\exists \vv{r'}. \left( \top_{\mmng{\Rel^+(\tau)}} \, \vv{T}\; \vv{z} \, s \land \bigwedge_{i \in [\ell]} P_i \right)\\
        &= \lambda s.\,\exists \vv{r'}. \left( 1 \land \bigwedge_{i \in [\ell]} P_i \right)\\
        &= \lambda s.\,1 \qquad \text{\ref{IH:nonempty_saf_0}}
\end{align*}

We have now established that $S(0,e'')$ holds for expressions all $e''$. 
The following case is the last remaining case to prove $S(n+1,b)$ for all base case expressions $b$:

\emph{Case $e = F:\sigma \in \calN$, for $n+1$.}
Thanks to~\ref{IH:nonempty_saf_BC_n+1}, $\bot\mhyphen\mathsf{free}_{\Rel^+(\iota)}(\mmng{\Delta_\HORS, \lift{\Gamma} \vdash \lift{\calR(F)} }(\beta^{n})\,\vv{z})$ and:
\begin{align*}
    \eqA 
        &= \mmng{\Delta_\HORS, \lift{\Gamma} \vdash \lift{F}}(\beta^{n+1}) \, \vv{z}\\
        &= \mmng{\Delta_\HORS, \lift{\Gamma} \vdash \lambda \vv{y} \, r.\,R_F\,\vv{y} \, r}(\beta^{n+1})\,\vv{z}\\
        &= \lambda s.\, \beta^{n+1,\overline{z},s}(R_F) \, \vv{z} \, s \\
        &= \lambda s.\, \beta^{n+1}(R_F) \, \vv{z} \, s \\
        &= \lambda s.\, \mmng{\Delta_\HORS, \lift{\Gamma} \vdash \lift{\calR(F)} }(\beta^{n}) \, \vv{z} \, s\\
        &= \mmng{\Delta_\HORS, \lift{\Gamma} \vdash \lift{\calR(F)} }(\beta^{n}) \, \vv{z}
\end{align*}

Finally, we present the remaining case to prove that $S(n',e'')$ for all $n' \geq 0$ and all expressions $e''$.

\emph{Case $e = F \, e_1 \dots e_\ell$ with $F \in \calN$ and $\ell>0$, for $n+1$.}
\begin{align*}
    \eqA 
        &= \lambda s.\,\exists \vv{r'}.\, \left( \mmng{\Delta_\HORS'' \vdash R_F}(\beta^{n+1,\overline{z},s}[\vv{r} \mapsto \vv{r'}]) \, \vv{T}\; \vv{z} \, s \land \bigwedge_{i \in [\ell]} P_i \right) \\
        &= \lambda s.\,\exists \vv{r'}.\, \left( \beta^{n+1,\overline{z},s}[\vv{r} \mapsto \vv{r'}](R_F) \, \vv{T}\; \vv{z} \, s \land \bigwedge_{i \in [\ell]} P_i \right)\\
        &= \lambda s.\,\exists \vv{r'}.\, \left( \beta^{n+1}(R_F) \, \vv{T}\; \vv{z} \, s \land \bigwedge_{i \in [\ell]} P_i \right)\\
        &= \lambda s.\,\exists \vv{r'}.\, \left( \mmng{\Delta_\HORS, \lift{\Gamma} \vdash \lift{\calR(F)}}(\beta^{n}) \, \vv{T}\; \vv{z} \, s \land \bigwedge_{i \in [\ell]} P_i \right)\\
        &\sqsupseteq \lambda s.\,  \mmng{\Delta_\HORS, \lift{\Gamma} \vdash \lift{\calR(F)}}(\beta^{n}) \, \vv{S}\; \vv{z} \, s \qquad \qquad \text{\ref{IH:nonempty_saf_IC_n+1}}\\
        &= \mmng{\Delta_\HORS, \lift{\Gamma} \vdash \lift{\calR(F)}}(\beta^{n}) \, \vv{S}\; \vv{z} 
\end{align*}
By \ref{IH:nonempty_saf_BC_n+1},
$\bot\mhyphen\mathsf{free}_{\Rel^+(\tau)}(\mmng{\Delta_\HORS, \lift{\Gamma} \vdash \lift{\calR(F)} }(\beta^{n}))$.
All arguments $\vv{S}$ and $\vv{z}$ are also $\bot$-free, so it follows that $\bot\mhyphen\mathsf{free}_{\Rel^+(\iota)}(\mmng{\Delta_\HORS, \lift{\Gamma} \vdash \lift{\calR(F)} }(\beta^{n}) \, \vv{S}\; \vv{z})$, as required.
\end{proof}

\begin{corollary}\label{cor:homc_cor_nonempty}
For all $n \geq 0$ and HoRS $\calG$,
$\bot\mhyphen\mathsf{free}_{\Rel^+(\iota)}(\mmng{\Delta_\calG \vdash \lift{S}}(\beta^n))$.
I.e.~there exists
a $\bot$-free tree $t \in \mmng{\iota}$ such that $\mmng{\Delta_\calG \vdash \lift{S}}(\beta^n)\,t$.
\end{corollary}

\correctnessNonemptiness*
\begin{proof}[Proof sketch]
Note that the constructed HoCHC logic program $\vdash P_\calG : \Delta_\calG$ is \emph{incremental} in the sense that each iteration of the one-step consequence operator (further) constrains a finite prefix of the trees it generates.
For sort $\iota \to o$, this means that either a contradiction occurs in finite time (e.g.~$\exists r_1.\,a\,r_1=r \land b\,r_1=r$ where $a,b$ are distinct unary alphabet symbols) or no contradiction occurs and the program is strictly incremental.

By Corollary~\ref{cor:homc_cor_nonempty}, no contradiction occurs after finite time.
This means that no contradiction occurs at all and $\vdash P_\calG : \Delta_\calG$ is strictly incremental.
It follows that there exists a $\bot$-free tree $t \in \mmng{\iota}$ such that $\mmng{\Delta_\calG \vdash \lift{S}}\left(\glb \beta^n\right)\,t$.
\end{proof}

\subsection{Inclusion}\label{sec:homc_corr_inclusion}

We aim to show that $\mmng{\Delta_\calG \vdash \lift{S}}(T_{P_\calG:\Delta_\calG}^{\calM \, n}(\top_{\Delta_\calG}))$ is included in $\imap_\iota(\hmng{\calN \vdash S}(\left(\hmng{\calG}_{\calN}^n(\bot_\calN) \right))$, for every $n \geq 0$ (Corollary~\ref{cor:correctness_inclusion}).
To this end, we define a family of logical relations in Definition~\ref{def:homc_inclusion_relations} to capture this notion at higher sorts and for larger sort environments, as proved in Lemma~\ref{lem:homc_cor_inclusion}.

The intuition is that the relation comprises pairs that preserve order on order-preserving arguments.

\begin{definition}
\label{def:homc_inclusion_relations}
We define a family of logical relations $\Incl_\sigma \subseteq \mmng{\sigma} \times \dmng{\sigma}$:
\begin{align*}
    \Incl_{\iota}(t_1,t_2) &\defeq t_1 = t_2\\
    \Incl_{\Rel^+(\iota)}(p_1,p_2) &\defeq p_1 \sqsubseteq p_2\\
    \Incl_{\Rel^+(\sigma_1\to\sigma_2)}(p_1,p_2) &\defeq \forall w \in \mmng{\Rel^-(\sigma_1)}.\, \forall z \in \hmng{\sigma_1}.\\
    &\qquad\qquad \Incl_{\Rel^-(\sigma_1)}(w,\imap^-_{\sigma_1}(z)) \Rightarrow \Incl_{\Rel^+(\sigma_2)}(p_1\,w, p_2\,\imap^-_{\sigma_1}(z))\\
    \Incl_{\Gamma}(\theta_\beta,\theta_\alpha)
        &\defeq \dom(\lift{\Gamma}) = \dom(\theta_\beta) \land \dom(\Gamma) = \dom(\theta_\alpha)\; \land \\
        &\qquad \bigwedge_{x:\sigma \in \Gamma} \Incl_{\Rel^-(\sigma)}(\theta_\beta(x'),\imap^-_\sigma(\theta_\alpha(x)))
\end{align*}
Alternatively, $\Incl_{\Rel^+(\sigma_1\to\dots\to \sigma_m \to \iota)}(p_1,p_2)$ can be defined as
\[
    \forall \vv{w}.\,\forall \vv{z}.\, \left(\bigwedge_{i \in [m]} \Incl_{\Rel^-(\sigma_i)}(w_i,\imap^-_{\sigma_i}(z_i))\right) \Rightarrow \Incl_{\Rel^+(\iota)}(p\,w_1 \dots w_m, p_2\,\imap^-_{\sigma_1}(z_1)\dots \imap^-_{\sigma_m}(z_m))
\]
for all $\vv{w}$ and $\vv{z}$ from the appropriate domains.
\end{definition}

Note that in general $\mmng{\sigma}$ differs from $\dmng{\sigma}$, and the arguments of $\Incl_\sigma$ do not necessarily live in the same set.
However, for the sort we are interested in, namely $\Rel^+(\iota) = \iota \to o$, the denotations $\mmng{\iota \to o}$ and $\dmng{\iota \to o}$ coincide (idem for $\iota$), so that the relations are well-defined.

\begin{lemma}\label{lem:homc_existential_inclusion}
For $F \in \hmng{\sigma_1 \to \dots \to \sigma_k \to \tau_1 \to \dots \to \tau_\ell \to \iota}$, $t_i \in \hmng{\sigma_i}$ for all $i \in [k]$, and $z_j \in \hmng{\tau_j}$ for all $j \in [\ell]$,
\[ 
    \left(\lambda s.\,\exists \vv{r}.\,F\,\vv{r}\,\vv{z} = s \land \bigwedge_{i \in [k]} t_i \sqsubseteq r_i \right)
    \sqsubseteq
    \left(\lambda s.\,\exists \vv{r}.\,F\,\vv{r}\,\vv{z} \sqsubseteq s \land \bigwedge_{i \in [k]} t_i \sqsubseteq r_i \right)
    \sqsubseteq 
    \left(\lambda s.\,F\,\vv{t}\,\vv{z} \sqsubseteq s \right)
\]
where $\vv{r} = r_1 \dots r_k$, $\vv{t} = t_1 \dots t_k$, and $\vv{z} = z_1 \dots z_\ell$.
\end{lemma}

\begin{lemma}\label{lem:homc_cor_inclusion}
For all $n \geq 0$, 
all typing judgements 
\(
\calN, \Gamma \vdash e : \sigma
\) of the HoRS $\calG$ where $\Gamma = \set{x_1 : \tau_1, \ldots, x_k : \tau_k}$,
and valuations $\theta_\alpha \in \hmng{\Gamma}$ and $\theta_\beta \in \mmng{\lift{\Gamma}}$,
\begin{align*}
    &\calN, \Gamma \vdash e:\sigma \land \Incl_{\Gamma}(\theta_\beta,\theta_\alpha) 
        \quad \Rightarrow \quad  \\
    & \qquad \Incl_{\Rel^+(\sigma)}\left(\mmng{\Delta_\calG, \lift{\Gamma} \vdash \lift{e}}(\beta^{n}), \imap_\sigma(\hmng{\calN,\Gamma \vdash e}(\alpha^{n}))\right)
\end{align*}
where 
\begin{align*}
\lift{\Gamma} &:= \set{x'_1: \Rel^-(\tau_1), \dots, x'_k: \Rel^-(\tau_k)}\\
\alpha^{n} &:= \left(\hmng{\calG}_{\calN}^n(\bot_\calN) \right) [\overline x \mapsto \theta_\alpha(x)] \in \hmng{\calN, \Gamma}\\ 
\beta^{n} &:= \left( T_{P_\calG:\Delta_\calG}^{\calM \, n}(\top_{\Delta_\calG})\right) [\overline{x'} \mapsto \vv{\theta_\beta(x')}] \in \mmng{\Delta_\calG, \lift{\Gamma}}.
\end{align*}
Notice that $\calN, \Gamma \vdash e : \sigma$ implies $\Delta_\calG, \lift{\Gamma} \vdash \lift{e} : \Rel^+(\sigma)$.
\end{lemma}
\begin{proof}
We proceed by induction on $n \geq 0$ within which (both in the base case and the induction step) we use structural induction on HoRS term $e$.
Some parts of the proof are presented out of order to avoid duplication.
In fact, the structure of this proof and order of presentation correspond to the proof of Lemma~\ref{lem:homc_botfree_lemma}, the structure of which is outlined in Figure~\ref{fig:homc_botfree_IH_overview}.
We again assume WLOG $e$ contains no $\lambda$s.

We use the following shorthands for $e:\sigma = \sigma_1 \to \dots \to \sigma_m \to \iota$, $w_i \in \mmng{\Rel^-(\sigma_i)}$, and $z_i \in \hmng{\sigma_i}$ such that $\Incl_{\Rel^-(\sigma_i)}(w_i,\imap^-_{\sigma_i}(z_i))$, for all $i \in [m]$:
\phantomsection\label{eq:homc_B_inclusion}\label{eq:homc_C_inclusion}\[
    B = \mmng{\Delta_\calG, \lift{\Gamma} \vdash \lift{e}}(\beta^{n})\,\vv{w}
    \qquad\qquad
    C = \imap_\sigma(\hmng{\calN,\Gamma \vdash e}(\alpha^{n}))\,\vv{\imap^-(z)}
\]
Thus, both $\eqB$ and $\eqC$ are both elements of $\mmng{\iota \to o} = \dmng{\iota \to o}$, and 
it suffices to show that $\eqB \sqsubseteq \eqC$.

We present the following base cases w.r.t. the structure of $e$ (also denoted $b$ for \emph{base case expression}).

\emph{Case $e = x: \iota \in \rsvars$.} 
For all $n \geq 0$, $\Incl_\Gamma(\theta_\beta,\theta_\alpha)$ implies $\beta^{n}(x') = \alpha^{n}(x)$ and:
\begin{align*}
    \eqB
        &= \mmng{\Delta_\HORS, \lift{\Gamma} \vdash \lift{x}}(\beta^{n})\\
        &= \mmng{\Delta_\HORS, \lift{\Gamma} \vdash \lambda r.\,x' = r}(\beta^{n})\\
        &= \lambda s \in \mmng{\iota}.\left(\beta^{n}(x') = s\right)\\
                &= \lambda s \in \mmng{\iota}.\left( \alpha^{n}(x) = s \right)\\
        &\sqsubseteq \imap_\iota(\alpha^{n}(x))\\
        &= \imap_\iota ( \hmng{\calN,\Gamma \vdash x}(\alpha^{n})) \\
        &= \eqC
\end{align*}

\emph{Case $e = \$: \sigma_1 \to \dots \to \sigma_m \to \iota \in \Sigma \cup \rsvars$.}
For all $n \geq 0$, we rely on Lemma~\ref{lem:homc_wellsorted_pair}.
If $\$ \in \dom(\rsvars)$, let $m>0$.
Then, $\Incl_{\Rel^-(\sigma)}(\beta^{n}(x'),\imap^-_{\sigma}(\alpha^{n}(x)))$ and $\Incl_{\Rel^-(\sigma_j)}(w_j,\imap^-_{\sigma_j}(z_j))$, for all $j \in [m]$, such that:

\vspace*{-10pt}
\hspace*{15pt}\begin{minipage}{.45\linewidth}
\begin{align*}
    \eqB
        &= \mmng{\Delta_\HORS, \lift{\Gamma} \vdash \lift{f}}(\beta^{n})\,\vv{w}\\
        &= \mmng{\Delta_\HORS, \lift{\Gamma} \vdash \lift{f}}(\beta^{n})\,\vv{z}\\
        &= \mmng{\Delta_\HORS, \lift{\Gamma} \vdash \lambda \vv{y} \, r.\,f\,\vv{y} = r}(\beta^{n})\,\vv{z}\\
        &= \lambda s \in \mmng{\iota}.\left(  \widehat{F_f} \, \vv{z} = s \right)\\
        &\sqsubseteq \imap_\iota(\widehat{F_f} \, \vv{z})\\
        &= \imap_\iota (\hmng{\calN,\Gamma \vdash f}(\alpha^{n}) \, \vv{z})\\
        &= \imap_{\sigma} (\hmng{\calN,\Gamma \vdash f}(\alpha^{n})) \, \vv{\imap^-(z)} \\
        &= \eqC
\end{align*}
\end{minipage} 
\begin{minipage}{.45\linewidth}
\begin{align*}
    \eqB
        &= \mmng{\Delta_\HORS, \lift{\Gamma} \vdash \lift{x}}(\beta^{n})\,  \vv{w}\\
        &= \mmng{\Delta_\HORS, \lift{\Gamma} \vdash \lambda \vv{y}\,r.\,x'\,\vv{y}\,r}(\beta^{n})\,  \vv{w}\\
        &= \lambda s \in \mmng{\iota}.\,\beta^{n} (x')\, \vv{w} \, s\\
        &= \beta^{n} (x')\, \vv{w}\\
        &\sqsubseteq \imap^-_{\sigma} (\alpha^{n}(x)) \, \vv{\imap^-(z)}\\
        &= \imap_{\sigma} (\hmng{\calN,\Gamma \vdash x}(\alpha^{n})) \, \vv{\imap^-(z)}\\
        &= \eqC\\
        \vphantom{\left(  \widehat{F_f} \, \vv{z} = s \right)} 
\end{align*}
\end{minipage}

This remaining base case expression is where we start needing induction on $n \geq 0$.

\emph{Case $e = F:\sigma \in \calN$, and $n = 0$.} 
We rely on Lemma~\ref{lem:homc_wellsorted_pair} to prove $\eqB$ and $\eqC$ are the universal relation:
\begin{align*}
    \eqB 
        &= \mmng{\Delta_\HORS, \lift{\Gamma} \vdash \lift{F}}(\beta^{0}) \, \vv{w}\\
        &= \mmng{\Delta_\HORS, \lift{\Gamma} \vdash \lambda \vv{y} \, r.\,R_F\,\vv{y} \, r}(\beta^{0})\,\vv{w}\\
                &= \lambda s \in \mmng{\iota}.\,\beta^{0}(R_F) \, \vv{w} \, s \\
                &= \top_{\mmng{\Rel^+(\iota)}} \\
                &= \top_{\dmng{\Rel^+(\sigma)}} \, \vv{\imap^-(z)} \\
        &= \imap_\sigma(\bot_{\hmng{\sigma}}) \, \vv{\imap^-(z)}\\
                &= \imap_\sigma(\hmng{\calN,\Gamma \vdash F}(\alpha^{0})) \, \vv{\imap^-(z)} \\
        &= \eqC
\end{align*}

This covers $n = 0$ for base case expressions.
Recall our proof follows the structure of Lemma~\ref{lem:homc_botfree_lemma}, which is outlined in Figure~\ref{fig:homc_botfree_IH_overview}.
We distinguish three induction hypotheses, where $S(n,e)$ denotes that the claim holds for $n$ and expression $e$.
Thus, we have now proved $S(0,b)$ for all base case expressions $b$. 

In this inductive case, we consider expressions $e = \$ \, e_1 \dots e_\ell : \sigma_1 \to \dots \to \sigma_m \to \iota$ for some $\ell > 0$.
For this, we introduce some more shorthands:
\[
    \Delta_\HORS' \defeq \Delta_\HORS, \lift{\Gamma}, \vv{y}, r \qquad
    \Delta_\HORS'' \defeq \Delta_\HORS, \lift{\Gamma},\vv{y}, r, \vv{r}\qquad
    \beta^{n,\vv{w},s} \defeq \beta^{n}[\vv{y} \mapsto \vv{w}, r \mapsto s]
\]

Note that the sort $\tau$ of $\$$ is of the form
\[ \tau = \tau_1 \to \dots \to \tau_\ell \to \sigma_1 \to \dots \to \sigma_m \to \iota \]
where $e_1 : \tau_1, \dots, e_\ell : \tau_\ell$, for some $\ell > 0$.
Sometimes we abbreviate $\sigma_1 \to \dots \to \sigma_m \to \iota$ to $\sigma$.

\phantomsection\label{eq:homc_dag_inclusion_proof}
In the sequel, steps marked with $\dag$ use \ref{IH:nonempty_saf_0} for $n=0$, and \ref{IH:nonempty_saf_IC_n+1} for $n>0$. 

Because $\eqB$ is $A$ from Lemma~\ref{lem:homc_botfree_lemma} with $\vv{w}$ substituted for (their) $\vv{z}$, we simply rewrite $\eqB$ like $A$ is rewritten \hyperref[eq:A_rewrite]{there};
for all $n \geq 0$ and expressions $e = \$ \, \vv{e}$, we can rewrite $\eqB$ to obtain:
\begin{align*}
    \eqB 
                                                                                &= \lambda s.\,\max\Big\{ \\
            & \qquad \qquad \qquad \min\{ \\
                & \qquad \qquad \qquad \qquad \,\mmng{\Delta_\HORS'' \vdash \$'}(\beta^{n,\overline{w},s}[\vv{r} \mapsto \vv{r'}]) ( \mmng{\Delta_\HORS'' \vdash \dlift{(e_1,r_1)}}(\beta^{n,\overline{w},s}[\vv{r} \mapsto \vv{r'}])) \\
                & \qquad \qquad \qquad \qquad \qquad \qquad \qquad \quad\quad\, \dots (\mmng{\Delta_\HORS'' \vdash \dlift{(e_\ell,r_\ell)}}(\beta^{n,\overline{w},s}[\vv{r} \mapsto \vv{r'}]))\, \vv{w} \, s, \\
                & \qquad \qquad \qquad \qquad \min\{\mmng{\Delta_\HORS'' \vdash Prop(e_i,r_i)}(\beta^{n,\overline{w},s}[\vv{r} \mapsto \vv{r'}]) \mid i \in [\ell]\} \\
            & \qquad \qquad \qquad \;\} \\
        & \qquad \quad  \mid \forall i \in [\ell].\, r'_i \in \mmng{\Rel^-(\tau_i)} \Big\}        \end{align*}
We now distinguish two cases for each subexpression $e_i:\tau_i$, namely $\tau_i = \iota$ and $\tau_i = \tau'_1 \to \tau'_2$.

If $e_i$ is of sort $\iota$, then the following holds:
\begin{align*}
    \mmng{\Delta_\HORS'' \vdash Prop(e_i:\iota,r_i)}(\beta^{n,\overline{w},s}[\vv{r} \mapsto \vv{r'}])
                        &= \mmng{\Delta_\HORS,\lift{\Gamma} \vdash \lift{e_i}}(\beta^{n})\,r_i'\\
    \mmng{\Delta_\HORS'' \vdash \dlift{(e_i:\iota,r_i)}}(\beta^{n,\overline{w},s}[\vv{r} \mapsto \vv{r'}])
                &= r_i'
\end{align*}
We know from \hyperref[eq:homc_dag_inclusion_proof]{$\dag$} that $\Incl_{\Rel^+(\iota)}(\mmng{\Delta_\HORS,\lift{\Gamma} \vdash \lift{e_i}}(\beta^{n}), \imap_\iota(\hmng{\calN,\Gamma \vdash e_i}(\alpha^{n}))$. 

Otherwise, in case $e_i:\tau_i = \tau'_1 \to \tau'_2$, the following holds:
\begin{align*}
    \mmng{\Delta_\HORS'' \vdash Prop(e_i:\tau_i,r_i)}(\beta^{n,\overline{w},s}[\vv{r} \mapsto \vv{r'}]) 
                &= 1  \\
    \mmng{\Delta_\HORS'' \vdash \dlift{(e_i:\tau_i,r_i)}}(\beta^{n,\overline{w},s}[\vv{r} \mapsto \vv{r'}]) 
                &= \mmng{\Delta_\HORS, \lift{\Gamma} \vdash \lift{e_i}}(\beta^{n})
\end{align*}
We know from \hyperref[eq:homc_dag_inclusion_proof]{$\dag$} that $\Incl_{\Rel^+(\tau_i)}(\mmng{\Delta_\HORS,\lift{\Gamma} \vdash \lift{e_i}}(\beta^{n}), \imap_{\tau_i}(\hmng{\calN,\Gamma \vdash e_i}(\alpha^{n}))$. 

As ``semantic equivalents'' of the above terms, let us write
\[
    P_i:o \defeq
        \left\{\begin{array}{ll}
            \mmng{\Delta_\HORS,\lift{\Gamma} \vdash \lift{e_i}}(\beta^{n})\,r_i' & \quad \text{if }\tau_i = \iota \\
            1 & \quad \text{if }\tau_i = \tau'_1 \to \tau'_2
        \end{array}\right.
\]
and
\[
    T_i:\Rel^-(\tau_i) \defeq
        \left\{\begin{array}{ll}
            r_i' & \quad \text{if }\tau_i = \iota \\
            \mmng{\Delta_\HORS, \lift{\Gamma} \vdash \lift{e_i}}(\beta^{n}) & \quad \text{if }\tau_i = \tau'_1 \to \tau'_2
        \end{array}\right.
\]
for all $i \in [\ell]$.
Additionally, we define
\[
    P_i':o \defeq
        \left\{\begin{array}{ll}
            \imap_\iota(\hmng{\calN,\Gamma \vdash e_i}(\alpha^{n}))\,r_i' & \quad \text{if }\tau_i = \iota \\
            1 & \quad \text{if }\tau_i = \tau'_1 \to \tau'_2
        \end{array}\right.
\]
and
\[
    T_i':\Rel^-(\tau_i) \defeq
        \left\{\begin{array}{ll}
            r_i' & \quad \text{if }\tau_i = \iota \\
            \imap_{\tau_i}(\hmng{\calN,\Gamma \vdash e_i}(\alpha^{n})) & \quad \text{if }\tau_i = \tau'_1 \to \tau'_2
        \end{array}\right.
\]
for all $i \in [\ell]$, to be used after applying the induction hypothesis \hyperref[eq:homc_dag_inclusion_proof]{$\dag$}. 
And finally, for all $i \in [\ell]$,
\[
    S_i:\Rel^-(\tau_i) \defeq
        \left\{\begin{array}{ll}
            r_i' & \quad \text{if }\tau_i = \iota \\
            \hmng{\calN,\Gamma \vdash e_i}(\alpha^{n}) & \quad \text{if }\tau_i = \tau'_1 \to \tau'_2
        \end{array}\right.
\]

We derive by abuse of notation, using the above:
\begin{align*}
    \eqB
        &= \lambda s.\,\exists \vv{r'}.\, \left( \mmng{\Delta_\HORS, \lift{\Gamma} \vdash \$'}(\beta^{n}) \, \vv{T}\; \vv{w} \, s \land \bigwedge_{i \in [\ell]} P_i \right)
\end{align*}

We continue by case analysis on $\$$. 

\emph{Case $e = f \, e_1 \dots e_\ell$ with $f \in \Sigma$.}
For all $n \geq 0$:
\begin{align*}
    \eqB 
        &= \lambda s.\,\exists \vv{r'}. \left( \mmng{\Delta_\HORS, \lift{\Gamma} \vdash D_f}(\beta^{n}) \, \vv{r'}\; \vv{w} \, s \land \bigwedge_{i \in [\ell]} \mmng{\Delta_\HORS,\lift{\Gamma} \vdash \lift{e_i}}(\beta^{n})\,r_i' \right)\\
        &= \lambda s.\,\exists \vv{r'}. \left( \widehat{F_f} \, \vv{r'}\; \vv{w} = s \land \bigwedge_{i \in [\ell]} \mmng{\Delta_\HORS,\lift{\Gamma} \vdash \lift{e_i}}(\beta^{n})\,r_i' \right) \\
        &\sqsubseteq \lambda s.\,\exists \vv{r'}. \left( \widehat{F_f} \, \vv{r'}\; \vv{w} = s \land \bigwedge_{i \in [\ell]} \imap_\iota(\hmng{\calN,\Gamma  \vdash e_i}(\alpha^{n}))\,r_i' \right)  \qquad  \hyperref[eq:homc_dag_inclusion_proof]{\dag}\\
        &= \lambda s.\,\exists \vv{r'}. \left( \widehat{F_f} \, \vv{r'}\; \vv{w} = s \land \bigwedge_{i \in [\ell]} \left(\hmng{\calN,\Gamma  \vdash e_i}(\alpha^{n}) \sqsubseteq \,r_i'\right) \right) \\
        &= \lambda s.\,\exists \vv{r'}. \left( \widehat{F_f} \, \vv{r'}\; \vv{z} = s \land \bigwedge_{i \in [\ell]} \left(\hmng{\calN,\Gamma  \vdash e_i}(\alpha^{n}) \sqsubseteq \,r_i'\right) \right) \\
        &\sqsubseteq \lambda s. \left(\widehat{F_f}\, \hmng{\calN,\Gamma \vdash e_1}(\alpha^{n}) \dots \hmng{\calN,\Gamma \vdash e_\ell}(\alpha^{n})\, \vv{z} \sqsubseteq s \right) \qquad \text{Lem~\ref{lem:homc_existential_inclusion}}\\
        &= \imap_\iota (\widehat{F_f} \, \hmng{\calN,\Gamma \vdash e_1}(\alpha^{n}) \dots \hmng{\calN,\Gamma \vdash e_\ell}(\alpha^{n})\, \vv{z}) \\
        &= \imap_\sigma (\hmng{\calN,\Gamma \vdash f}(\alpha^{n})) \, \imap^-_\iota(\hmng{\calN,\Gamma \vdash e_1}(\alpha^{n})) \dots \imap^-_\iota(\hmng{\calN,\Gamma \vdash e_\ell}(\alpha^{n}))\, \vv{z} \\
        &= \eqC
\end{align*}

\emph{Case $e = x \, e_1 \dots e_\ell$ with $x \in \rsvars$ of sort $\tau = \tau_1 \to \tau_2$.}
For all $n \geq 0$: 
\begin{align*}
    \eqB 
        &= \lambda s.\,\exists \vv{r'}.\left( \mmng{\Delta_\HORS, \lift{\Gamma} \vdash x'}(\beta^{n}) \, \vv{T}\; \vv{w} \, s \land \bigwedge_{i \in [\ell]} P_i \right)\\
        &= \lambda s.\,\exists \vv{r'}.\, \left( \beta^{n}(x') \, \vv{T}\; \vv{w} \, s \land \bigwedge_{i \in [\ell]} P_i \right)\\
        &\sqsubseteq \lambda s.\,\exists \vv{r'}. \left( \imap_\tau(\alpha^{n}(x)) \, \vv{T'}\; \vv{\imap^-(z)} \, s \land \bigwedge_{i \in [\ell]} P_i \right) \qquad \hyperref[eq:homc_dag_inclusion_proof]{\dag}\\
        &\sqsubseteq \lambda s.\,\exists \vv{r'}. \left( \imap_\tau(\alpha^{n}(x)) \, \vv{T'}\; \vv{\imap^-(z)} \, s \land \bigwedge_{i \in [\ell]} P'_i \right) \qquad \hyperref[eq:homc_dag_inclusion_proof]{\dag}\\
        &= \lambda s.\,\exists \vv{r'}. \left( \imap_\iota(\alpha^{n}(x) \, \vv{S}\; \vv{z}) \, s \land \bigwedge_{i \in [\ell]} P'_i \right) \qquad \text{Lem~\ref{lem:homc_wellsorted_pair}}\\
        &= \lambda s.\,\exists \vv{r'}. \left( \hmng{\calN,\Gamma \vdash x}(\alpha^{n}) \, \vv{S}\; \vv{z} \sqsubseteq s \land \bigwedge_{i \in [\ell]} P_i' \right)\\
        &\sqsubseteq \lambda s.\left(\hmng{\calN,\Gamma \vdash x\,\vv{e}}(\alpha^{n})\, \vv{z} \sqsubseteq s\right) \qquad \text{Lem~\ref{lem:homc_existential_inclusion}}\\
        &= \imap_\iota (\hmng{\calN,\Gamma \vdash x\,\vv{e}}(\alpha^{n})\, \vv{z})\\
        &= \imap_\tau (\hmng{\calN,\Gamma \vdash x\,\vv{e}}(\alpha^{n}))\, \vv{\imap^-(z)} \qquad \text{Lem~\ref{lem:homc_wellsorted_pair}}\\
        &= \eqC
\end{align*}
Recall that $\Incl_{\Rel^-(\tau)}(\beta^{n}(x'),\imap^-_\tau(\alpha^{n}(x)))$. 
The IH \hyperref[eq:homc_dag_inclusion_proof]{$\dag$} gives us $\Incl_{\Rel^-(\tau_i)}(T_i,T_i')$.
Because we also have $\Incl_{\Rel^-(\sigma_j)}(w_j,\imap^-_{\sigma_j}(z_j))$, we derive the first inclusion.

\emph{Case $e = F \, e_1 \dots e_\ell$ with $F:\sigma \in \calN$, and $n=0$.} 
\begin{align*}
    \eqB 
        &= \lambda s.\,\exists \vv{r'}. \left( \mmng{\Delta_\HORS, \lift{\Gamma} \vdash R_F}(\beta^{n}) \, \vv{T}\; \vv{w} \, s \land \bigwedge_{i \in [\ell]} P_i \right)\\
                                &\sqsubseteq \lambda s.\,1\\
        &= \imap_\iota(\bot) \qquad \text{Lem~\ref{lem:homc_wellsorted_pair}}\\
        &= \imap_\iota(\bot_{\hmng{\sigma}} \, \hmng{\calN,\Gamma \vdash e_1}(\alpha^{0}) \dots \hmng{\calN,\Gamma \vdash e_\ell}(\alpha^{0}) \, \vv{z})\\
        &= \imap_\iota(\alpha^{0}(F) \, \hmng{\calN,\Gamma \vdash e_1}(\alpha^{0}) \dots \hmng{\calN,\Gamma \vdash e_\ell}(\alpha^{0}) \, \vv{z})\\
        &= \imap_\iota(\hmng{\calN,\Gamma \vdash F}(\alpha^{0}) \, \hmng{\calN,\Gamma \vdash e_1}(\alpha^{0}) \dots \hmng{\calN,\Gamma \vdash e_\ell}(\alpha^{0}) \, \vv{z})\\
        &= \imap_\sigma(\hmng{\calN,\Gamma \vdash F}(\alpha^{0})) \, \imap^-_{\tau_1}(\hmng{\calN,\Gamma \vdash e_1}(\alpha^{0})) \dots \imap^-_{\tau_\ell}(\hmng{\calN,\Gamma \vdash e_\ell}(\alpha^{0})) \, \vv{\imap^-(z)}\\
        &= \eqC
\end{align*}

We have now established that $S(0,e'')$ holds for expressions all $e''$. 
The following case is the last remaining case to prove that $S(n+1,b)$ for all base case expressions $b$:

\emph{Case $e = F:\sigma \in \calN$, for $n+1$.} 
\begin{align*}
    \eqB 
        &= \mmng{\Delta_\HORS, \lift{\Gamma} \vdash \lift{F}}(\beta^{n+1}) \, \vv{w}\\
        &= \mmng{\Delta_\HORS, \lift{\Gamma} \vdash \lambda \vv{y} \, r.\,R_F\,\vv{y} \, r}(\beta^{n+1})\,\vv{w}\\
                        &= \lambda s.\, \beta^{n+1}(R_F) \, \vv{w} \, s \\
        &= \lambda s.\, \mmng{\Delta_\HORS, \lift{\Gamma} \vdash \lift{\calR(F)} }(\beta^{n}) \, \vv{w} \, s \\
        &\sqsubseteq \lambda s.\, \imap_\sigma(\hmng{\calN,\Gamma \vdash \calR(F)}(\alpha^{n})) \, \vv{\imap^-(z)} \, s \quad \text{\ref{IH:nonempty_saf_BC_n+1}} \\
        &= \lambda s.\, \imap_\sigma(\alpha^{n+1}(F)) \, \vv{\imap^-(z)} \, s \\
                        &= \imap_\sigma(\hmng{\calN, \Gamma \vdash F}(\alpha^{n+1})) \, \vv{\imap^-(z)} \\
        &= \eqC
\end{align*}
Finally, we present the remaining case to prove that $S(n',e'')$ for all $n' \geq 0$ and all expressions $e''$.

\emph{Case $e = F \, e_1 \dots e_\ell$ with $F \in \calN$ and $\ell>0$, for $n+1$.}
\begin{align*}
    \eqB 
        &= \lambda s.\,\exists \vv{r'}. \left( \mmng{\Delta_\HORS, \lift{\Gamma} \vdash R_F}(\beta^{n+1}) \, \vv{T}\; \vv{w} \, s \land \bigwedge_{i \in [\ell]} P_i \right)\\
        &= \lambda s.\,\exists \vv{r'}.\, \left( \beta^{n+1}(R_F) \, \vv{T}\; \vv{w} \, s \land \bigwedge_{i \in [\ell]} P_i \right)\\
        &= \lambda s.\,\exists \vv{r'}. \left( \mmng{\Delta_\HORS, \lift{\Gamma} \vdash \lift{\calR(F)}}(\beta^{n}) \, \vv{T}\; \vv{w} \, s \land \bigwedge_{i \in [\ell]} P_i \right)\\
        &\sqsubseteq \lambda s.\,\exists \vv{r'}. \left( \imap_\tau(\hmng{\calN,\Gamma \vdash \calR(F)}(\alpha^{n})) \, \vv{T'}\; \vv{\imap^-(z)} \, s \land \bigwedge_{i \in [\ell]} P_i \right) \qquad \text{\ref{IH:nonempty_saf_BC_n+1}, }\hyperref[eq:homc_dag_nonemptiness_proof]{\dag} \\
        &\sqsubseteq \lambda s.\,\exists \vv{r'}. \left( \imap_\tau(\hmng{\calN,\Gamma \vdash \calR(F)}(\alpha^{n})) \, \vv{T'}\; \vv{\imap^-(z)} \, s \land \bigwedge_{i \in [\ell]} P'_i \right) \qquad \hyperref[eq:homc_dag_nonemptiness_proof]{\dag} \\
        &= \lambda s.\,\exists \vv{r'}. \left( \imap_\iota(\hmng{\calN,\Gamma \vdash \calR(F)}(\alpha^{n}) \, \vv{S}\; \vv{z}) \, s \land \bigwedge_{i \in [\ell]} P_i' \right) \qquad \text{Lem~\ref{lem:homc_wellsorted_pair}}\\
                &= \lambda s.\,\exists \vv{r'}. \left( \alpha^{n+1}(F) \, \vv{S}\; \vv{z} \sqsubseteq s \land \bigwedge_{i \in [\ell]} P_i' \right)\\
        &= \lambda s.\,\exists \vv{r'}. \left( \hmng{\calN,\Gamma \vdash F}(\alpha^{n+1}) \, \vv{S}\; \vv{z} \sqsubseteq s \land \bigwedge_{i \in [\ell]} P_i' \right)\\
        &\sqsubseteq \imap_\iota (\hmng{\calN,\Gamma \vdash  F\,\vv{e}}(\alpha^{n+1})\,\vv{z}) \qquad \text{Lem~\ref{lem:homc_existential_inclusion}}\\
        &= \imap_\tau (\hmng{\calN,\Gamma \vdash  F\,\vv{e}}(\alpha^{n+1}))\,\vv{\imap^-(z)} \qquad \text{Lem~\ref{lem:homc_wellsorted_pair}}\\
        &= \eqC
\end{align*}
\end{proof}

\correctnessInclusion*

\subsection{Main result: equality}\label{sec:homc_corr_finalised}

\auxiliaryDescendingChain*
\begin{proof}
Recall that $\mmng{\Delta_\calG}$ and $\mmng{\rho}$ are complete lattices for each relational sort environment $\Delta_\calG$ and relational sort $\rho$. 
Thus, we know that the greatest lower bounds exist.

To perform induction on the structure of $e:\sigma$, we strengthen the claim to the following.

For all typing judgements 
\(
\calN, \Gamma \vdash e : \sigma
\) 
of the HoRS $\calG$ where $\Gamma = \set{x_1 : \tau_1, \ldots, x_k : \tau_k}$,
for all descending chains of valuations $\calI \subseteq \mmng{\Delta_\calG}$,
and valuations $\theta \in \mmng{\lift{\Gamma}}$,
\begin{align*} 
    &\mmng{\Delta_\calG, \lift{\Gamma} \vdash \lift{e: \sigma}}\left(\left(\glb \calI\right)[\overline{x'} \mapsto \vv{\theta(x')}]\right)\\
        &= \glb_{I \in \calI} \mmng{\Delta_\calG, \lift{\Gamma} \vdash \lift{e: \sigma}}(I[\overline{x'} \mapsto \vv{\theta(x')}]).
\end{align*}
where $\lift{\Gamma} = \set{x'_1: \Rel^-(\tau_1), \dots, x'_k: \Rel^-(\tau_k)}$.
We abbreviate $I[\overline{x'} \mapsto \vv{\theta(x')}]$ to $I'$.
Note that
\[ 
\left(\glb \calI\right)[\overline{x'} \mapsto \vv{\theta(x')}] 
=
\glb_{I \in \calI} I[\overline{x'} \mapsto \vv{\theta(x')}]
=
\glb_{I \in \calI} I',
\]
so that we shorten the above equation to
\[
    \mmng{\Delta_\calG, \lift{\Gamma} \vdash \lift{e: \sigma}}\left(\glb_{I \in \calI} I'\right)
        = \glb_{I \in \calI} \mmng{\Delta_\calG, \lift{\Gamma} \vdash \lift{e: \sigma}}(I').
\]

\emph{Case $e = f: \iota^k \to \iota \in \Sigma$.}
The meaning of $\lift{f}$ is independent of the valuation, as demonstrated by:
\[
\lift{f} = \lambda y_1 \dots y_k\,r.\,(f\,y_1 \dots y_k = r)
\]

\emph{Case $x:\sigma \in \rsvars$.}
The meaning of $\lift{x}$ relies only on the $\theta$ part of the valuation, as evident from:
\[
    \lift{x:\sigma} 
        = \left\{\begin{array}{ll}
        \lambda r.\,(x' = r) & \quad \text{if }\sigma = \iota\\
        \lambda y_1 \dots y_k \, r.\,x'\,y_1\dots y_k\,r & \quad \text{if }\sigma = \sigma_1 \to \dots \to \sigma_k \to \iota \text{ for }k>0
        \end{array}\right.
\]

\emph{Case $F\sigma_1 \to \dots \to \sigma_k \to \iota \in \calN$.}
\begin{align*}
    \mmng{\Delta_\calG, \lift{\Gamma} \vdash \lift{F}}\left(\glb_{I \in \calI} I'\right)
        &= \mmng{\Delta_\calG, \lift{\Gamma} \vdash \lambda y_1\dots y_k\,r.\,R_F\,y_1\dots y_k\,r}\left(\glb_{I \in \calI} I'\right)\\
        &= \lambda y_1\dots y_k\,r.\,\left(\glb_{I \in \calI} I'\right)(R_F)\, y_1\dots y_k\,r\\
        &= \lambda y_1\dots y_k\,r.\,\left(\glb_{I \in \calI} I\right)(R_F)\, y_1\dots y_k\,r\\
        &= \lambda y_1\dots y_k\,r.\,\left(\glb_{I \in \calI} I(R_F)\right)\, y_1\dots y_k\,r \\
        &= \glb_{I \in \calI} I(R_F)\\
        &= \glb_{I \in \calI} I'(R_F)\\
        &= \glb_{I \in \calI} \mmng{\Delta_\calG, \lift{\Gamma} \vdash R_F}(I')\\
        &= \glb_{I \in \calI} \mmng{\Delta_\calG, \lift{\Gamma} \vdash \lambda y_1\dots y_k\,r.\,R_F\,y_1\dots y_k\,r}(I')\\
        &= \glb_{I \in \calI} \mmng{\Delta_\calG, \lift{\Gamma} \vdash \lift{F}}(I')
\end{align*}
For the fourth equality, we rely on the codomain of $\calI$ being a complete lattice (namely, a finite product of complete lattices $\mmng{\rho}$).

\emph{Case $e = \$\,\vv{e}$ with $\vv{e} = e_1\dots e_\ell$ for $\ell>0$.}
This case follows from applying the induction hypothesis in a straightforward though laborious unfolding of the relational lift and the semantics.
Recall that:
\begin{align*}
    &\mmng{\Delta_\HORS, \lift{\Gamma} \vdash \lift{\$\,\vv{e}}}\left(\glb_{I \in \calI} I'\right) \\
        &= \mmng{\Delta_\HORS, \lift{\Gamma} \vdash \lambda \vv{y}\,r.\, \exists \vv{r}.\, \$' \, \dlift{(e_1,r_1)}\dots \dlift{(e_\ell,r_\ell)}\, \vv{y} \, r \land \bigwedge_{i \in [\ell]} Prop(e_i,r_i)}\left(\glb_{I \in \calI} I'\right)
\end{align*}
The previous cases show that the greatest lower bound is preserved by $\mmng{\$'}$.
Observe that $\dlift{(e_i,r_i)}$ is $r_i$ or $\lift{e_i}$. 
Either way, the greatest lower bound is preserved by $\mmng{\dlift{(e_i,r_i)}}$.
Similarly, $Prop(e_i,r_i)$ is either $\lift{e_i}\,r_i$ or  $\mathsf{true}$, and the greatest lower bound is thus preserved by $\mmng{Prop(e_i,r_i)}$.
This concludes the proof.
\end{proof}

\correctnessEquality*

\end{document}